\documentclass[11pt]{article}
\usepackage{rotating,amsmath,amssymb,amsfonts,amsthm,longtable,caption}
\usepackage{epsfig,mathrsfs,natbib,bm,color,graphics,multirow,makecell,booktabs}
\def\boxit#1{\vbox{\hrule\hbox{\vrule\kern6pt \vbox{\kern6pt#1\kern5pt}
\kern6pt\vrule}\hrule}}

\newcommand{\bX}{{\boldsymbol X}}

\newcommand{\by}{{\boldsymbol y}}
\newcommand{\bY}{{\boldsymbol Y}}
\newcommand{\bx}{{\boldsymbol x}}
\newcommand{\bz}{{\boldsymbol z}}
\newcommand{\bZ}{{\boldsymbol Z}}

\newcommand{\bA}{{\boldsymbol A}}

\newcommand{\bE}{{\boldsymbol E}}

\newcommand{\bV}{{\boldsymbol V}}
\newcommand{\bG}{{\boldsymbol G}}
\newcommand{\bU}{{\boldsymbol U}}

\newcommand{\bI}{{\boldsymbol I}}
\newcommand{\bJ}{{\boldsymbol J}}
\newcommand{\bM}{{\boldsymbol M}}
\newcommand{\bs}{{\boldsymbol s}}
\newcommand{\bS}{{\boldsymbol S}}
\newcommand{\bt}{{\boldsymbol t}}

\newcommand{\mR}{\mathbb{R}}

\newcommand{\var}{{\mbox{Var}}}

\newcommand{\bbeta}{{\boldsymbol \beta}}

\newcommand{\bepsilon}{{\boldsymbol \epsilon}}

\newcommand{\bSigma}{{\boldsymbol \Sigma}}

\newtheorem{theorem}{Theorem}[]
\newtheorem{lemma}{Lemma}[]

\newtheorem{algorithm}{Algorithm}[]
\newtheorem{remark}{Remark}[]
\begin{document}
\bibliographystyle{asa}

\title{Markov Neighborhood Regression for High-Dimensional Inference}

\author{Faming Liang, Jingnan Xue, Bochao Jia\thanks{ Liang is Professor, Department of Statistics, Purdue University, West Lafayette, IN 47906,
 email: fmliang@purdue.edu. Xue is Data Scientist at Houzz, Palo Alto, California. 
 Jia is Research Scientist, Eli Lilly and Company, Lilly Corporate Center, Indianapolis, IN 46285.}
}


\maketitle

\begin{abstract}

 This paper proposes an innovative method
 for constructing confidence intervals and assessing $p$-values in statistical inference for high-dimensional linear 
 models. The proposed method has successfully broken
 the high-dimensional inference problem into a series of low-dimensional inference problems: 
 For each regression coefficient $\beta_i$, the confidence interval and $p$-value are computed 
 by regressing on a subset of variables selected according to the 
 conditional independence relations between the corresponding variable $X_i$ 
 and other variables. Since the subset of variables 
 forms a Markov neighborhood of $X_i$ in 
 the Markov network formed by all the variables $X_1,X_2,\ldots,X_p$, 
 the proposed method is coined as Markov neighborhood regression.  The proposed method is tested on high-dimensional linear, 
 logistic and Cox regression. The numerical results indicate that the proposed method 
 significantly outperforms the existing ones. 
 Based on the Markov neighborhood regression, a method of learning causal structures for high-dimensional linear 
 models is proposed and applied to  
 identification of drug sensitive genes
 and cancer driver genes. 
 The idea of using conditional independence relations for dimension
 reduction is general and potentially can be extended to 
 other high-dimensional or big data problems as well.  

\vspace{1mm}

\end{abstract}


{\centering \section{Introduction}} 

 During the past two decades, the dramatic improvement in data collection and acquisition technologies
 has enabled scientists to collect a great amount of high-dimensional data,
 for which the dimension $p$ can be much larger than the sample size $n$ (a.k.a. small-$n$-large-$p$).
 The current research on high-dimensional data mainly focuses on
 variable selection and graphical modeling. 
 The former aims to find a consistent estimate for high-dimensional regression under a
 sparsity constraint. The existing methods include
  Lasso \citep{Tibshirani1996}, SCAD \citep{FanL2001}, MCP \citep{Zhang2010}, and
 rLasso \citep{SongL2015a}, among others.
  The latter aims to learn conditional independence relationships for a large set of
  variables. The existing methods include  graphical Lasso \citep{YuanL2007, FriedmanHT2008},
 nodewise regression \citep{MeinshausenB2006}, and $\psi$-learning \citep{LiangSQ2015}, among others.

 Recently many researchers turn their attention to statistical inference for high-dimensional data, 
 aiming to quantify uncertainty of high-dimensional regression, e.g., constructing confidence intervals and assessing $p$-values for the regression coefficients.
 The Bayesian methods \citep{LiangSY2013, SongL2015b} can potentially be used for this purpose,
 but are time-consuming. The existing frequentist methods include  
 desparsified Lasso, multi sample-splitting, and ridge projection, among others.   
 Refer to \cite{DezeureBMM2015} for a comprehensive overview.
 
 The desparsified Lasso method was proposed in \cite{vandeGeer2014}, which is also essentially
 the same as the one developed in \cite{ZhangZhang2014} and \cite{Javanmard2014}.
 For the high-dimensional linear regression 
\[
\bY=\bX \bbeta+\bepsilon,
\]
 where $\bepsilon$ are zero-mean Gaussian random errors, desparsified Lasso defines a bias-corrected 
 estimator
 \begin{equation} \label{dlassoeq0}
 \hat{\bbeta}_{bc}=\hat{\bbeta}_{Lasso}+\hat{\Theta}\bX^{T}(\by-\bX\hat{\bbeta}_{Lasso})/n,
 \end{equation} 
 where $\hat{\bbeta}_{Lasso}$ is the original Lasso estimator, and $\hat{\Theta}$ is
 an approximator to the inverse of $\hat{\bSigma}=\bX^{T}\bX/n$. 
 From (\ref{dlassoeq0}), one can obtain
\begin{equation} \label{dlassoeq}
 \sqrt{n}(\hat{\bbeta}_{bc}-\bbeta)=\hat{\Theta}\bX^{T}\bepsilon/\sqrt{n}
  +\sqrt{n}(I_{p}-\hat{\Theta} \hat{\bSigma})(\hat{\bbeta}_{Lasso}-\bbeta):= 
   \hat{\Theta}\bX^{T}\bepsilon/\sqrt{n}+\Delta_n,
\end{equation}
where $I_{p}$ denotes the $p \times p$ identity matrix, and $\Delta_n$ is the error term.
With an appropriate estimator $\hat{\Theta}$, e.g., the one obtained by nodewise regression \citep{meinshausen2006},
it can be shown that 
 $\|\Delta_n\|_{\infty}=o_{p}(1)$ and thus 
  $\sqrt{n}(\hat{\bbeta}_{bc}-\bbeta)$ shares the same asymptotic distribution
 with $\hat{\Theta}\bX^{T}\bepsilon/\sqrt{n}$.
Further, to calculate confidence intervals for $\bbeta$, one needs to
approximate the distribution of $\hat{\Theta}\bX^{T}\bepsilon/\sqrt{n}$.
For example,  \cite{Javanmard2014} approximated it by $N(0,\hat{\sigma}^{2}\hat{\Theta}\hat{\bSigma}\hat{\Theta}^{T})$,
where $\hat{\sigma}^2$ is a consistent estimator of $\sigma^2$; 
and \cite{ZhangandCheng} approximated it using multiplier bootstrap.

 The multi sample-splitting method was proposed and analyzed in \cite{MeinshausenMB2009}, which  
 works in the following procedure:  Splitting the samples into two subsets equally, 
 using the first half of samples for variable selection 
 and using the second half of samples for calculating $p$-values based on the selected variables; repeating this process for many times; and aggregating the 
 $p$-values for statistical inference.  
 The confidence intervals can be constructed based on their duality with $p$-values. 
 The idea about sample-splitting and subsequent 
 statistical inference has also been implicitly contained in \cite{WassermanR2009}.   
 The multi sample-splitting method is very general 
 and can be applied to many different types of models.  
 The ridge projection method was studied in \cite{Buhlmann2013}, which can be viewed as 
 a direct extension of the low-dimensional ridge regression
 to the high-dimensional case.  The bias of the ridge estimator has been assessed and corrected,  
 the distribution of the ridge estimator has been derived and approximated, 
 and thus the method can be used for statistical inference.  
 
 The other methods include residual-type bootstrapping \citep{ChatterjeeL2013, LiuYu2013},  
 covariance test \citep{Lockhart2014}, and group-bound \citep{Meinshausen2015}.  
 A problem with the residual-type bootstrapping method is the super-efficiency 
 phenomenon; that is, a confidence interval of a zero regression coefficient might be the 
 singleton $\{0\}$. The covariance test method relies on the solution path of the Lasso 
 and is much related to the post-selection 
 inference methods \citep{Berk2013, LeeDSYJ2016, Tibshirani2016, FithianST2014}.  
 The group-bound method provides a good treatment for highly correlated variables, but 
 has often a weak power in detecting the effect of individual variables. 
 Recently, some methods based on the idea of estimating a low-dimensional component 
 of a high-dimensional model have also been proposed, see e.g. \cite{Belloni2015}
 and \cite{YangYun2017}.

  This paper proposes an innovative method, the so-called Markov neighborhood regression (MNR), for high-dimensional inference.  
  By making use of conditional independence relations among different explanatory variables, 
  it successfully breaks the high-dimensional inference problem into a series of low-dimensional 
  inference problems. The proposed method is fundamentally different from the existing ones, such as 
  desparsified-Lasso, ridge projection, and multi sample-splitting, 
  which strive to work on the problem in its original scale. The proposed method has been tested on high-dimensional linear, 
  logistic and Cox regression. The numerical results indicate that the proposed method 
  significantly outperforms the existing ones, while having competitive computational efficiency. Based on the concept of MNR, this paper also proposes a new method for learning the 
  causal structure of high-dimensional regression and applies it to identification of drug sensitive genes and cancer driver genes. 
  The idea of using conditional independence for dimension
  reduction is general and can be applied to 
  many other high-dimensional or big data problems as well.  

  The remaining part of this paper is organized as follows. Section 2 describes the MNR method and establishes its validity. 
  Section 3 presents numerical results  
  on simulated data along with comparisons with some existing methods. 
  Section 4 proposes a method for learning the
  causal structures of high-dimensional regression. 
  Section 5 presents some real data examples.
  Section 6 concludes the paper with a brief discussion. 

\section{Markov Neighborhood Regression} 

\subsection{Linear Regression} 

 Suppose that a set of $n$ independent samples $D_{n}=\{(Y^{(i)},\bX^{(i)})_{i=1}^{n} \}$ have been collected
 from the linear regression with a random design: 
\begin{equation} \label{modeleq1}
 Y=\beta_0+X_{1}\beta_1+ \ldots+ X_{p}\beta_p+\epsilon,
\end{equation}
 where $\epsilon$ follows the normal distribution $N(0,\sigma^{2})$,
 and the explanatory variables, or known as features, 
 $\bX=(X_{1}, \ldots, X_{p})$ follows a multivariate normal distribution $N_{p}(\mathbf{0},\Sigma)$. In what follows, we will call 
 $\{X_i: \beta_i\ne 0, i=1,2,\ldots,p \}$ and  $\{X_i: \beta_i=0, i=1,2,\ldots,p\}$ 
 the true and false features, respectively. 
 Without loss of generality, we assume that the features have been standardized
 such that $E(X_{j})=0$ and Var$(X_{j})=1$ for $j=1,\dots,p$.
 Further, we represent $\bX$ by an undirected graph $\bG=(\bV,\bE)$,
 where $\bV=\{1,2,\ldots,p\}$ denotes the set of $p$ vertices, $\bE=(e_{ij})$ denotes the adjacency matrix,
 $e_{ij}=1$ if the $(i,j)$th entry of the precision matrix $\Theta=\Sigma^{-1}$ is nonzero and
 0 otherwise.
  We use $\bX_{\bA}=\{X_k: k\in \bA\}$ to denote a set of features
  indexed by $\bA \subset \bV$,
  and use $P_{\bV}$ to denote the joint probability distribution of $\bX_{\bV}$.
  For a triplet $\bI, \bJ, \bU \subset \bV$,
  we use $\bX_{\bI} \perp \bX_{\bJ} | \bX_{\bU}$ to denote that $\bX_{\bI}$ is {\it conditionally independent}  of $\bX_{\bJ}$ given $\bX_{\bU}$. 
  A {\it path} of length $l>0$ from a vertex $v_0$ to another vertex $v_l$ is a sequence 
  $v_0,v_1,\ldots,v_l$ of distinct vertices such that $e_{v_{k-1},v_k}=1$ for $k=1,2,\ldots,l$.
  The subset $\bU \subset \bV$ is said to {\it separate}  $\bI \subset \bV$ from
 $\bJ \subset \bV$ if for every $i \in \bI$ and $j \in \bJ$, all paths from
 $i$ to $j$ have at least one vertex in $\bU$.
   $P_{\bV}$ is said to satisfy the {\it Markov property} with respect to
 $\bG$ if for every triple of disjoint sets $\bI, \bJ, \bU \subset \bV$, it holds that $X_{\bI} \perp X_{\bJ} | X_{\bU}$
  whenever $\bU$ separates $\bI$ and $\bJ$ in $\bG$.
  Let $\xi_{j}=\{k: e_{jk}=1\}$
  denote the neighboring set of $X_{j}$ in $\bG$.
  Following from the Markov property of the Gaussian graphical model (GGM), we have
  $X_{j} \perp X_{i} |\bX_{\xi_j}$ for any $i \in \bV \setminus \xi_j$, as $\xi_j$ forms a separator
  between $X_i$ and $X_j$.
 For convenience, we call $\xi_j$ the minimum Markov neighborhood of $X_{j}$ in $\bG$,
 and call any superset $\bs_j \supset \xi_j$ a Markov neighborhood of $X_{j}$ in $\bG$.
 The minimum Markov neighborhood is also termed as Markov blanket in Bayesian networks or 
 general Markov networks. 

 To motivate the proposed method, we first look at a simple mathematical fact based on the 
 well known property of Gaussian graphical models (see e.g. \cite{Lauritzen1996}): 
 \[
  X_i \perp X_j |X_{\bV\setminus \{i,j\}} \Longleftrightarrow \theta_{ij}=0,
  \]
where $\theta_{ij}$ denotes the $(i,j)$-th entry of $\Theta$.  
 Without loss of generality, we let $\bS_1=\{2, \ldots, d\}$
 denote a Markov neighborhood of $X_1$,  
 let $\Sigma_d$ denote the covariance matrix of $\{X_{1}\} \cup \bX_{\bS_1}$,
 and partition $\Theta$ as
\begin{equation} \label{Leq1}
\Theta= \begin{bmatrix}
\Theta_{d}     & \Theta_{d,p-d} \\
\Theta_{p-d,d}       & \Theta_{p-d}  
\end{bmatrix},
\end{equation}
where the first row of $\Theta_{d,p-d}$ and the first column of $\Theta_{p-d,d}$ are exactly zero, as 
$X_1 \perp \bX_{\bV\setminus(\{1\}\cup \bS_1)} | \bX_{\bS_1}$ holds. 
Inverting $\Theta$, we have
$\Sigma_{d}=(\Theta_{d}-\Theta_{d,p-d}\Theta_{p-d}^{-1}\Theta_{p-d,d})^{-1}$,
 which is equal to the top $d\times d$-submatrix of $\Theta^{-1}$. Therefore, 
\begin{equation} \label{Leq2}
 \Sigma_{d}^{-1}=\Theta_{d}-\Theta_{d,p-d}\Theta_{p-d}^{-1}\Theta_{p-d,d}.
\end{equation}
 Since the first row of $\Theta_{d,p-d}$ and the first column of $\Theta_{p-d,d}$ are exactly zero,
 the $(1,1)$th element of  $\Theta_{d,p-d}\Theta_{p-d}^{-1}\Theta_{p-d,d}$ is exactly zero. Therefore,
   the $(1,1)$-th entry of $\Theta_d$ (and $\Theta$) equals to the $(1,1)$-th entry of $\Sigma_{d}^{-1}$.
 This suggests that {\it if $\{X_{1}\} \cup \bX_{\bS_1} \supset \bX_{\bS_*}$ holds and 
 $n$ is sufficiently large,  where $\bS_*$ denote the index 
 set of true features of the model (\ref{modeleq1}), then the statistical inference for
 $\beta_1$ can be made based on the 
 subset regression:}
 \begin{equation} \label{modeleq2}
 Y=\beta_0'+X_{1}\beta_1+X_{2} \beta_2'+\ldots+X_{d} \beta_d'+\epsilon,  
 \end{equation}
 where the prime on $\beta_k$'s for $k\ne 1$
 indicates that those regression coefficients might be modified by the subset regression. 
 Since $\bS_1$ forms a Markov neighborhood of $X_{1}$ in the Markov network formed 
 by all features, we call (\ref{modeleq2})
 a {\it Markov neighborhood regression}, which breaks the high-dimensional inference problem
 into a series of low-dimensional inference problems.
 Based on this mathematical fact, we propose the following algorithm:
 
 \begin{algorithm} (Markov Neighborhood Regression) \label{subsetAlg1}
 \begin{itemize}
 \item[(a)] (Variable selection) Conduct variable selection for the model (\ref{modeleq1})
  to get a consistent estimate of $\bS_*$.  Denote the estimate by $\hat{\bS_*}$.
 
 \item[(b)] ({\it Markov blanket estimation})
   Construct a GGM for $\bX$ and
   obtain a consistent estimate of the Markov blanket for each variable.
   Denote the estimates by $\hat{\xi}_j$ for $j=1,2,\ldots, p$.

 \item[(c)] (Subset regression) For each variable $X_j$, $j=1,\ldots,p$, 
  let $D_j=\{j\}\cup \hat{\xi}_j \cup \hat{\bS}_*$ and run an Ordinary Least Square (OLS) regression with
  the features given by $\bX_{D_j}$, i.e., 
 \begin{equation} \label{regeqD}
 Y=\beta_0+\bX_{D_j}\bbeta_{D_j}+\epsilon,
 \end{equation}
  where $\epsilon \sim N(0,\sigma^2I_n)$ and $I_n$ is an $n\times n$-identity matrix.
  Conduct inference for $\beta_j$, including the estimate, confidence interval and $p$-value,
  based on the output of (\ref{regeqD}). 
\end{itemize}
\end{algorithm}

  The Markov neighborhood corresponding to 
  the subset regression (\ref{modeleq2}) is 
  $\{2,3,\ldots,d\} \supseteq \hat{\xi}_1 \cup \hat{\bS}_*$. 
  In general, $\hat{\xi}_1 \cup \hat{\bS}_*$ can be any a subset of $\{1,2,\ldots,p\}$ 
  depending on the ordering of features in (\ref{modeleq1}).   
 Algorithm \ref{subsetAlg1} can be implemented in many different ways. For example,
 a variety of high-dimensional variable selection algorithms can be
 used for step (a), e.g.,
 Lasso \citep{Tibshirani1996}, SCAD \citep{FanL2001}, MCP \citep{Zhang2010} and
 rLasso \citep{SongL2015a}, which are all able to produce a consistent
 estimate for $\bS_*$ under appropriate conditions. Similarly, 
 a variety of graphical model learning algorithms can be used for step (b),
  e.g., graphical Lasso \citep{YuanL2007, FriedmanHT2008},
 nodewise regression \citep{MeinshausenB2006}, and $\psi$-learning \citep{LiangSQ2015},
 which all produce a consistent estimate for the underlying GGM under
 appropriate conditions.
 To justify Algorithm \ref{subsetAlg1}, we introduce the following lemma, which can be proved based on the basic theory of high-dimensional linear regression and the mathematical relation shown around equations (\ref{Leq1}) and (\ref{Leq2}). Refer to the Supplementary Material for the detail of the proof. 

\begin{lemma} \label{lemma01}
 Let $\hat{\xi}_{j}\supseteq\xi_{j}$ denote any Markov neighborhood of feature $x_{j}$,  let 
 $\hat{\bS}_* \supseteq \bS_*$ denote any reduced feature space, and let $D_j=\{j\} \cup \hat{\xi}_{j} \cup \hat{\bS}_*$.  
 Consider the subset regression (\ref{regeqD}).
 Let $\hat{\bbeta}_{D_j}$ denote the OLS estimator of $\bbeta_{D_j}$ from the subset regression, and   
 let $\hat{\beta}_{j}$ denote the element of $\hat{\bbeta}_{D_j}$ corresponding to the variable $X_j$. 
 If $|D_j|=o(n^{1/2})$, as $n\to\infty$, the following results hold:
 \begin{itemize}
 \item[(i)] $\sqrt{n}(\hat{\beta}_{j}-\beta_{j}) \sim 
 N(0,\sigma^{2}\theta_{jj})$,  where 
   $\theta_{jj}$ is the $(j,j)$-th entry of the precision matrix $\Theta$.  
 \item[(ii)] 
 $\sqrt{n} \frac{\hat{\beta}_j-\beta_j}{\sqrt{ \hat{\sigma}_n^2 \hat{\theta}_{jj}}}  \sim N(0,1)$,  
  where $\hat{\sigma}_n^2=(\by-\bx_{D_j}\hat{\bbeta}_{D_j})^T(\by-\bx_{D_j}\hat{\bbeta}_{D_j})/(n-d-1)$, 
  $\hat{\theta}_{jj}$ is the $(j,j)$-th entry of 
  the matrix $\bigg[\frac{1}{n} \sum_{i=1}^{n}\bx^{(i)}_{D_j}(\bx^{(i)}_{D_j})^{T} \bigg]^{-1}$, and
  $\bx^{(i)}_{D_j}$ denotes the $i$-th row of $\bX_{D_j}$.  
\end{itemize}
\end{lemma} 
 
 \begin{remark} \label{remA} 
  For the case that $n$ is finite, we have   
  $(n-|D_j|-1) \hat{\sigma}_n^2/\sigma^2 \sim \chi^2(n-|D_j|-1)$, independent of $\hat{\bbeta}_{D_j}$, 
  by the standard theory of OLS estimation.  
  Therefore, we can use $t(n-|D_j|-1)$ to approximate the distribution of 
  $\sqrt{n} \frac{\hat{\beta}_j-\beta_j}{\sqrt{ \hat{\sigma}_n^2 \hat{\theta}_{jj}}}$; 
  that is, {\it the estimate, $p$-value and confidence interval of $\beta_j$ can be 
  calculated from (\ref{regeqD}) as in conventional low-dimensional multiple linear regression. }
 \end{remark}

 Lemma \ref{lemma01} implies that Algorithm \ref{subsetAlg1} will be valid 
 as long as the following conditions hold:
\begin{eqnarray} 
 \hat{\xi}_{j} & \supseteq & \xi_{j}, \ \forall j \in \{1,2,\ldots,p\}, \label{ceq1} \\
 \hat{\bS}_* & \supseteq &  \bS_*, \label{ceq2} \\
 |D_j|&= & o(\sqrt{n}). \label{ceq3} 
\end{eqnarray} 
 Condition (\ref{ceq2}) is the so-called screening property, which is known to be satisfied by many 
 high-dimensional variable selection algorithms, such as SCAD \citep{FanL2001}, MCP \citep{Zhang2010}
 and adaptive Lasso \citep{Zou2006}.  Lasso also satisfies this condition if the design matrix satisfies the 
 compatibility condition \citep{vandeGeerB2009}, $|\bS_*|=o(n/\log(p))$, and the beta-min condition holds. 
 See \cite{DezeureBMM2015} for more discussions on this issue. 
 Given the sure screening property of the above variable selection procedure, if the nodewise regression algorithm \citep{MeinshausenB2006} is applied to learn the GGM in step (b) of Algorithm \ref{subsetAlg1}, then 
 the condition (\ref{ceq1}) can be satisfied. In fact, as along as the GGM construction algorithm 
 is consistent, the condition (\ref{ceq1}) will be asymptotically satisfied. 
 Further, the condition (\ref{ceq3}) can be easily 
 satisfied by a slight twist of the sparsity conditions used in the variable selection and 
 GGM estimation algorithms.  

 As an example, we give in the Appendix a set of technical conditions (A0)--(A9) 
 under which the conditions (\ref{ceq1})--(\ref{ceq3}) can be asymptotically satisfied, 
 provided that the SCAD algorithm is used for variable selection and 
 the $\psi$-learning algorithm is used for GGM estimation. 
 Based on these technical conditions, 
 the validity of Algorithm \ref{subsetAlg1} is justified in Theorem \ref{Them1}, whose proof is straightforward based on Slutsky's theorem and some existing theoretical results. Refer to 
 the Supplementary Material for the detail. 
 If different algorithms are used in Algorithm \ref{subsetAlg1}, then the 
 conditions used in Theorem \ref{Them1} should be changed accordingly.  
 We note that many conditions we imposed in proving the theorem
 are purely technical and only serve to provide theoretical understanding 
 of the proposed method. We have no intent to make the conditions the weakest possible.  

\begin{theorem} \label{Them1} (Validity of Algorithm \ref{subsetAlg1}) 
   If the conditions (A0)-(A9) hold, 
   the SCAD algorithm is used for variable selection in step (a), 
   and the $\psi$-learning algorithm is used for GGM construction in step (b),  
  then for each $j \in \{1,2,\ldots,p_n\}$, we have 
  $\sqrt{n} \frac{\hat{\beta}_j-\beta_j}{\sqrt{ \hat{\sigma}_n^2 \hat{\theta}_{jj}}}  \sim N(0,1)$ as $n\to\infty$, 
  where $\hat{\beta}_j$ denotes the estimate of $\beta_j$ obtained from the subset regression, 
  $\hat{\sigma}_n^2=(\by-\bx_{D_j}\hat{\bbeta}_{D_j})^T(\by-\bx_{D_j}\hat{\bbeta}_{D_j})/(n-d-1)$,
  $\hat{\theta}_{jj}$ is the $(j,j)$-th entry of
  the matrix $\bigg[\frac{1}{n} \sum_{i=1}^{n}\bx^{(i)}_{D_j}(\bx^{(i)}_{D_j})^{T} \bigg]^{-1}$, and
  $\bx^{(i)}_{D_j}$ denotes the $i$-th row of $\bX_{D_j}$.
\end{theorem}  

\begin{remark} \label{remB}   Following Remark \ref{remA}, we can conduct inference for $\beta_j$  based on the output of the subset regression (\ref{regeqD}) as in 
  conventional low-dimensional multiple linear regression.  
 \end{remark}
 
 As implied by Theorem \ref{Them1}, variable selection for regression (\ref{modeleq1})
 can be converted as a multiple hypothesis testing problem for simultaneously testing the hypotheses
 \begin{equation} \label{mult-test}
 H_{0,j}:\  \beta_j=0 \Longleftrightarrow  H_{1,j}: \ \beta_j \ne 0, \quad  j=1,2,\ldots,p,
 \end{equation}
 based on the $p$-values obtained from the subset regressions. The consistency of 
 this test-based method follows from Theorem 2 of \cite{LiangSQ2015} as discussed in 
 Section \ref{CausalGaussoian}. 
 Compared to the regularization methods, e.g. Lasso, MCP, SCAD and adaptive Lasso \cite{Zou2006}, a significant 
 advantage of this method is that it controls the false discovery rate
 of selected features in an explicit way. In addition, since the screening property generally holds for these 
 regularization methods, see \cite{DezeureBMM2015}  for discussions on this issue, 
 the new method might result in a lower false discovery rate as shown 
 in Table \ref{FDRtab}.
 On the other hand,  since the $p$-value measures the contribution of 
 a feature to the regression 
 conditioned on all other $p-1$ features, MNR might not work well  
 when strong collinearity exists among certain true and false features. This case has been 
 excluded by Conditions A2 and A4, where the fixed upper bounds
 on correlations and $\psi$-partial correlations place some 
 additional restrictions on the design. 
  
  The essential conditions required by MNR 
  are only sparsity; that is, the true regression model is sparse and 
  the conditional independence structure among the features is sparse such that (\ref{ceq1})-(\ref{ceq3}) hold when appropriate algorithms are applied. 
 Similar conditions have been assumed by some existing  methods. For example, desparsified-Lasso requires the true model to be of size  
  $o(\sqrt{n}/\log p)$ (see e.g., Fact 2 of \cite{DezeureBMM2015}), which is 
  a little more restrictive than $o(n^{1/2})$ required by MNR;
 desparsified-Lasso also requires the precision matrix $\Theta$ to be row-sparse at a level 
 of $o(n/\log p)$, which is comparable with the Markov blanket size $o(\sqrt{n})$ required by MNR 
 when $\log(p)=n^{\delta}$ for some $\delta \approx 1/2$.  
 The multi sample-splitting and the 
 ridge projection methods require the screening property (\ref{ceq2}) only, 
 which seems weaker than the conditions required by MNR and desparsified-Lasso. 
 However, as shown later by numerical examples, 
 they both essentially fail even for the simple linear regression case. 
 The use of conditional independence relations seems important for high-dimensional inference. 
 
  For MNR, since the essential conditions
  are (\ref{ceq1})-(\ref{ceq3}),
  a variable screening-based algorithm will also work under appropriate conditions. Based on this observation, we propose Algorithm S1, which together with some numerical results are presented in the Supplementary Material. Compared to Algorithm \ref{subsetAlg1}, Algorithm S1 can be substantially faster but the resulting confidence intervals can be a little wider; that is, Algorithm S1 is an accuracy/efficiency trade-off version of Algorithm \ref{subsetAlg1}. 

Finally, we note that there are indeed scenarios that conditions (\ref{ceq1})-(\ref{ceq3})
are violated. For example, if  all the features are equally correlated or there  are a few features whose Markov blanket is of size $O(\sqrt{n})$ or larger, then the condition (\ref{ceq1}) will be violated, as the algorithm always restricts the Markov blanket to be of size $o(\sqrt{n})$ or smaller. Similarly, if the true model is of size 
 $O(\sqrt{n})$ or larger, then condition (\ref{ceq2}) will be violated. These conditions can also be violated by the algorithms used for Markov blanket estimation or variable selection, particularly when the sample size is small. 
 The screening property is itself a large sample property. Our numerical experience shows that the MNR method is pretty robust to violations of the conditions (\ref{ceq1})-(\ref{ceq3}). This will be demonstrated in Section \ref{equicSect}.

\subsection{Generalized Linear Models}
 
 The MNR method can be easily extended to the generalized linear models (GLMs) 
 whose density function is given in the canonical form
\begin{equation} \label{GLMeq1}
f(y|\bx,\bbeta)=\exp(\vartheta y-b(\vartheta)+c(y)),
\end{equation}
where $b(\cdot)$ is continuously differentiable, and $\vartheta$ is the natural parameter relating $y$ to the features $\bx$ via a linear function 
$\vartheta=\beta_0+x_{1}\beta_{1}+\cdots+x_{p} \beta_{p}$.
 This class of GLMs includes Poisson regression, logistic regression and linear regression (with known variance).
 Note that for Cox proportional hazards models, the parameters can be estimated by maximizing the 
 partial likelihood function \citep{Cox1975}, based on which the Cox regression
 can be converted to a Poisson regression. See, e.g., Chapter 13 
 of \cite{McCu:Neld:1989} for the detail. 
 This conversion is important, which enables the use of the MNR method for
 high-dimensional survival data analysis.
 
 To justify this extension, we establish the following lemma, where we assume that the features follow a multivariate normal distribution $N(0,\Sigma)$ and each has been standardized to have a mean 0 and variance 1. The proof follows the same logic as that of Lemma 
 \ref{lemma01}, but the precision matrix involved in Lemma \ref{lemma01} is replaced by the inverse of the Fisher information matrix of the GLM. Refer to the Supplementary Material for the detail.  

\begin{lemma}\label{GLMlemma}
 Let  $\hat{\xi}_{j}\supseteq \xi_{j}$ denote any Markov neighborhood of feature $x_{j}$, 
 let  $\hat{\bS}_* \supseteq \bS_*$ denote any reduced feature space, and let 
  $D_j=\{j\} \cup \hat{\xi}_{j} \cup \hat{\bS}_*$. 
 Consider a subset GLM with the features $\bX_{D_j}$, let $\hat{\bbeta}_{D_j}$ 
 denote the MLE of $\bbeta_{D_j}$, and let $\hat{\beta}_{j}$ denote the 
 component of $\hat{\bbeta}_{D_j}$ corresponding to feature $X_j$. 
 If $|D_j|=o(n^{1/2})$, then, as $n\to \infty$, the following results hold:
 \begin{itemize}
 \item[(i)]  $\sqrt{n}(\hat{\beta}_{j}-\beta_{j}) \sim N(0,k_{jj})$, 
  where $k_{jj}$ denotes the $(j,j)$-th entry of the inverse of the Fisher information matrix 
  $K=I^{-1}=[E(b''(\bx^{T}\bbeta)\bx\bx^{T})]^{-1}$,
  and $\bbeta$ denotes the true regression coefficients.  
 \item[(ii)] $\sqrt{n}(\hat{\beta}_{j}-\beta_{j})/\sqrt{\hat{k}_{jj}} \sim N(0,1)$, where 
 $\hat{k}_{jj}$ denotes the $(j,j)$-th entry of the inverse of the observed information matrix 
 $J_n(\hat{\bbeta}_{D_j})=-\sum_{i=1}^n H_{\hat{\bbeta}_{D_j}}(\log f(y_i|\bbeta_{D_j}, \bx_{D_j}))/n$ 
 and $H_{\hat{\bbeta}_{D_j}}(\cdot)$ 
 denotes the Hessian matrix evaluated at the MLE $\hat{\bbeta}_{D_j}$.
 \end{itemize}
 \end{lemma}
 
Lemma \ref{GLMlemma} implies that the estimate, $p$-value and confidence interval of $\beta_j$ can be calculated from 
the subset GLM as in conventional low-dimensional GLMs. 
For GLMs, variable selection can be done using the SCAD, MCP or Lasso algorithm, and 
variable screening can be done using the sure independence screening algorithm 
developed in \cite{FanS2010}. By Theorem 5 of \cite{FanS2010}, we can bound the size 
of $\hat{\bS}_*$ by $O(n^{\frac{1}{2}-\frac{\varepsilon}{2}})$ for a small constant $\varepsilon>0$
with a slight modification of the technical conditions therein.  
Therefore, the theorems parallel to Theorem \ref{Them1} and Theorem S1 (in the Supplementary Material) can be proved for GLMs. For simplicity, they are omitted in the paper. 
 
\subsection{Joint Inference} \label{Sect2.3}

 The MNR method described above deals with only one coefficient
 $\beta_j$ in each subset regression. 
 In fact, it can be easily extended to conduct joint inference for
 several coefficients. 
 Let $\bA\subset \bV$ denote a set of features for which the joint inference
for the corresponding coefficients is desired. 
Define $\xi_{\bA}=\cup_{j \in \bA} \xi_j$ as the union
 of the Markov blankets of the features in $\bA$.
 Let $\bM=\bA \cup \hat{\xi}_{\bA} \cup \hat{\bS}_*$.
 Then a subset regression can be conducted with the features included in $\bM$.
 For high-dimensional linear regression, if $|\bM|=O(n^{1/2})$, then, similar to Theorem 1, we can show
 $\sqrt{n} (\hat{\bbeta}_A- \bbeta_A) \sim N(0, \sigma^2 \Theta_{AA})$,
 where $\Theta_{AA}$ denotes the submatrix of the precision matrix $\Theta$ constructed
  by its $A$ rows and $A$ columns.
  Similarly, for high-dimensional GLMs, we can show 
  $\sqrt{n} (\hat{\bbeta}_A- \bbeta_A) \sim N(0, K_{AA})$,
  where $K_{AA}$ denotes the submatrix of $K=[ E(b''(\bx^T \bbeta) \bx \bx^T]^{-1}$
  constructed by its $A$ rows and $A$ columns.

\section{Simulation Studies} 
 
 \subsection{A Conceptual Experiment} \label{Sect3.1}
  We first test the concept of MNR using a large-$n$-small-$p$ example; that is, whether  
  the confidence intervals generated by MNR 
  coincide with those generated by the OLS method as the sample size 
  $n$ becomes large. 
  We generated a dataset from the model (\ref{modeleq1}) with $n=2000$ and
 $p$=50, where $\sigma^2$ was set to 1, 
 the covariates $\bX$ were generated from a zero-mean multivariate Gaussian distribution with 
 a Toeplitz covariance matrix given by 
 $\Sigma_{i,j}=0.9^{|i-j|}$ for $i,j=1,\ldots,p$, 
 and the true regression coefficients $(\beta_0,\beta_1,\beta_2,\ldots,\beta_5)=(1,0.2,0.4,-0.3,-0.5,1.0)$ and
 $\beta_6=\cdots=\beta_p=0$.  

 Figure \ref{intervalplot} compares the 95\% confidence intervals of $\beta_1,\ldots,\beta_p$
 produced by MNR and OLS with the simulated dataset.
 For MNR, the nodewise regression algorithm
 (with SIS-Lasso performed for each node)
 was employed for Markov blanket estimation, and SIS-SCAD was employed for variable selection.  
 The SIS-Lasso refers to a variable selection procedure implemented in the package {\it SIS} 
 \citep{SISpackage}, 
 where the sure independence screening (SIS) algorithm \citep{FanLv2008} was first applied for variable screening
 and then the Lasso algorithm was applied to select variables from those survived from the screening procedure.  The SIS-SCAD and SIS-MCP
 can be interpreted in the same way. 
 As expected from Theorem \ref{Them1}, OLS and MNR produced almost identical
 confidence intervals for each regression coefficient.
 In this simulation, we set $n$ excessively large, which ensures
 the convergence of the sample covariance matrix to
 the true covariance matrix $\Theta^{-1}$. 
 In fact, MNR can work well with a smaller value of $n$ as illustrated 
 by the following small $n$-large-$p$ examples.   

\begin{figure}[htbp]
\centering
\begin{center}
\begin{tabular}{c}
\epsfig{figure=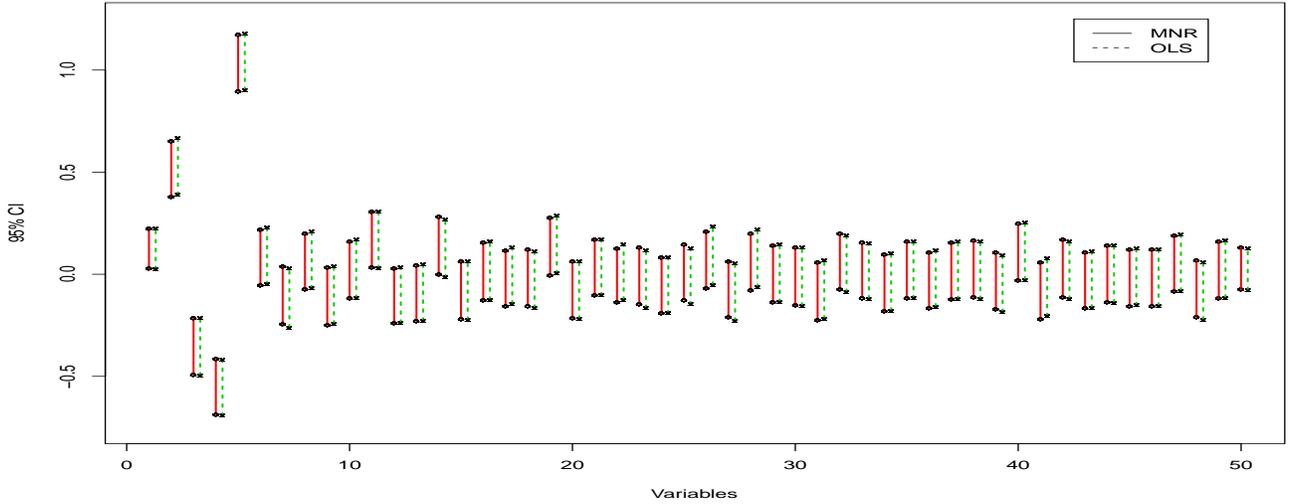,height=7.0in,width=3.0in,angle=270}
\end{tabular}
\end{center}
\caption{ The 95\% confidence intervals of $\beta_1,\ldots,\beta_p$
 produced by the MNR (solid line) and OLS (dashed line) methods for a dataset with $n=2000$ and $p=50$.}
\label{intervalplot}
\end{figure}
  

\subsection{An Illustrative Example} 

To illustrate the performance of MNR, we generated 100 independent datasets from
 the regression (\ref{modeleq1}), where  $n=200$, $p=500$, $\sigma^2=1$,   
 the features were generated from a zero-mean multivariate Gaussian distribution with 
 a Toeplitz covariance matrix given by  
 $\Sigma_{i,j}=0.9^{|i-j|}$ for $i,j=1,\ldots,p$, and  
 the true regression coefficients were given by 
 $(\beta_0,\beta_1,\beta_2,\ldots,\beta_5)=(1,2,4,-3,-5,10)$ and
 $\beta_6=\cdots=\beta_p=0$. We note that the same covariance matrix has 
 been used in \cite{vandeGeer2014} to illustrate the performance of desparsified-Lasso. 
 For convenience, we call this model a Toeplitz-covariance linear regression model.
 
 \subsubsection{Illustration of MNR}
 
 Algorithm \ref{subsetAlg1} was run for this example as in Section \ref{Sect3.1}, i.e., 
 applying SIS-SCAD for variable selection and the nodewise regression algorithm
 for Markov blanket estimation. 
 Table \ref{reg0tab} summarizes the coverage rates and widths
 of the 95\% confidence intervals produced by MNR for each regression coefficient. 
 For the non-zero regression coefficients (denoted by ``signal''), the 
 mean coverage rate and mean width of the confidence intervals are defined, respectively, by 
 \begin{equation} \label{measureeq}
 \bar{p}_{\rm cover} = \sum_{j=1}^{100}\sum_{i \in \bS_*} \hat{p}_i^{(j)}/(100\cdot |\bS_*|), 
 \quad 
 \bar{w}_{\rm CI}=\sum_{j=1}^{100} \sum_{i\in \bS_*} \hat{w}_i^{(j)}/(100\cdot |\bS_*|),
 \end{equation} 
 and their respective standard deviations are defined by 
 \begin{equation} \label{measureeq2}
 \begin{split}
 \sigma(\bar{p}_{\rm cover}) & =\sqrt{\var\{\hat{p}_i^{(j)}: i \in \bS_*, j=1,2,\ldots,100\}/100}, \\
 \sigma(\bar{w}_{\rm CI}) & =\sqrt{\var\{\hat{w}_i^{(j)}: i\in\bS_*, j=1,2,\ldots,100\} /100}, \\
 \end{split}
 \end{equation}
where $\hat{w}_i^{(j)}$ denotes the width of the 95\% confidence interval 
of $\beta_i$ constructed with the $j$th dataset, $\hat{p}_i^{(j)}\in \{0,1\}$ indicates the coverage of $\beta_i$ by the confidence interval, 
and $\var\{\cdot\}$ denotes the variance. By dividing by 100 in (\ref{measureeq2}), 
 the standard deviation represents the variability of the mean value (averaged over 100 independent 
datasets) for a single regression coefficient.
For the zero regression coefficients (denoted by ``noise''), 
the mean coverage rate, the mean width, and their standard deviations can be defined similarly.

 For comparison, we applied the
 desparsified Lasso, ridge projection and multi-split methods to this example. 
 These methods have been implemented in the $R$ package {\it hdi} \citep{Meierhdi2016}.
 The comparison indicates that MNR significantly outperforms the existing methods:
 for both the non-zero and zero regression coefficients, 
 the mean coverage rates produced by MNR are much closer to their nominal level. 
 The reason why desparsified-Lasso suffers from coverage deficiency for non-zero 
 regression coefficients will be explained in Section \ref{LRSect}. 
 
 \begin{table}[htbp] 
\caption{Coverage rates and widths of the 95\% confidence intervals 
produced by MNR with Algorithm \ref{subsetAlg1} for
the Toeplitz-covariance linear regression model,  where ``signal'' and ``noise'' denote non-zero and zero regression coefficients, respectively. For ``signal'', the reported mean value and standard deviation (in the parentheses) are
defined in (\ref{measureeq}) and (\ref{measureeq2}), respectively. For ``noise'', they are defined similarly. }
\vspace{-0.2in}
\label{reg0tab}
\begin{center}
\begin{tabular}{cccccc} \toprule
  Measure   &        &  Desparsified-Lasso  &  Ridge  & Multi-Split   & MNR   \\ \midrule
  & signal  &  0.384(0.049) & 0.576(0.049) &  0.202(0.040) &  0.956(0.021)  \\
 \raisebox{1.5ex}{Coverage} & noise  &  0.965(0.018) & 0.990(0.010) &  1.000(6.4e-4) 
   &  0.950(0.022)  \\ \midrule
  & signal  &   0.673(0.005) & 1.086(0.010) &  2.711(0.097) &  0.822(0.011)  \\
\raisebox{1.5ex}{Width}  & noise  &  0.691(0.005) & 1.143(0.008) &  2.790(0.103) &  0.869(0.007)  \\ \bottomrule
\end{tabular}
\end{center}
\end{table}

As discussed previously, MNR converts the problem of variable selection 
as a multiple hypothesis testing problem. To illustrate the potential of MNR 
in variable selection, we converted the $p$-values produced by the 
subset regressions to z-scores using the inverse 
probability integral transformation 
\begin{equation} \label{zscoreq}
Z_i^{(j)}=\Phi^{-1}(1-q_i^{(j)}), \quad \quad i=1,2,\ldots,p, \quad j=1,2,\ldots,100,
\end{equation}
where $q_i^{(j)}$ denotes the $p$-value calculated via 
subset regression for feature $i$ with dataset $j$, and $\Phi(\cdot)$ denotes the CDF 
of the standard Gaussian distribution. 
Figure S1 (in the Supplementary Material) shows the histogram of the z-scores, which indicates that the true and false features can be
well separated by the z-scores. The empirical Bayesian method developed by \cite{LiangZ2008} was
applied to each of the 100 datasets for simultaneously 
testing the hypotheses (\ref{mult-test}). At a FDR level of $q=0.0001$, which 
is measured by the q-value of \cite{Storey2002}, the method led to exact 
identifications of the true and false features for all 100 datasets, i.e., both the  
false selection rate (FSR) and negative selection rate (FSR) are 0. 
More results were shown in Table \ref{FDRtab}. 
Here the FSR and NSR are defined by 
\[
FSR=\frac{\sum_{j=1}^{100}|\hat{\bS}_{j}\setminus \bS_*|}{\sum_{j=1}^{100}|\hat{\bS}_{j}|}, \quad \quad 
NSR=\frac{\sum_{j=1}^{100}|\bS_*\backslash\hat{\bS}_{j}|}{\sum_{j=1}^{100}|\bS_*|},
\]
where $\bS_*$ is the set of true features, and $\hat{\bS}_{j}$ is the set
of selected features for dataset $j$.
For comparison, SIS-SCAD, SIS-MCP and SIS-Lasso were applied to these datasets for 
performing variable selection under their default settings in the package {\it SIS}. 
Table \ref{FDRtab} shows that MNR can significantly outperform the existing methods 
in high-dimensional variable selection. As mentioned previously,  compared to the existing 
methods, a significant advantage of the MNR-based variable selection 
method is that it controls the FDR of selected features.  

 \begin{table}[htbp] 
\caption{Variable selection for the Toeplitz-covariance linear regression 
with the MNR, SIS-SCAD, SIS-MCP and SIS-Lasso methods.} 
\vspace{-0.2in}
\label{FDRtab}
\begin{center}
\begin{tabular}{ccccccc} \toprule
            &  \multicolumn{3}{c}{MNR} &   &   &  \\ \cline{2-4}  
  \raisebox{1.5ex}{Measure}   &  $q=0.0001$ & $q=0.001$ & $q=0.01$ & 
  \raisebox{1.5ex}{SIS-SCAD}   &  \raisebox{1.5ex}{SIS-MCP}  & \raisebox{1.5ex}{SIS-Lasso} \\ \midrule
  FSR       &   0   &  0.004    &  0.022 &  0.127   &  0.175    & 0.819  \\
  NSR       &   0   &  0   &  0   &  0  &  0  &  0 \\ \bottomrule
\end{tabular}
\end{center}
\vspace{-0.25in}
\end{table}

\subsubsection{Illustration of Joint Inference with MNR}

To illustrate the use of MNR for joint inference, we 
constructed Bonferroni joint confidence intervals based on 
the subset regression for each of the following sets of parameters: 
 $(\beta_1,\beta_2)$, $(\beta_3,\beta_4,\beta_5)$, $(\beta_1,\beta_6)$,
$(\beta_7, \beta_{10})$, and $(\beta_{20},\beta_{200},\beta_{400})$, 
which have covered the cases of combinations of nonzero coefficients, 
combinations of zero and nonzero coefficients, and combinations of zero coefficients.  
For each set of parameters, as described in Section 
\ref{Sect2.3}, the subset regression was constructed by unioning the Markov neighborhoods of 
the corresponding features, and then the 95\% joint confidence intervals
for the set of parameters were constructed 
using the standard Bonferroni method. The Markov neighborhood of each feature was constructed 
as in Section \ref{Sect3.1} using nodewise regression for GGM estimation and 
SIS-SCAD for variable selection. Table \ref{jointtab} summarizes the coverage rates 
of the joint confidence intervals, and it  
indicates that the proposed method works reasonably well for this example. 

 \begin{table}[htbp] 
\caption{Coverage rates of 
 95\% joint confidence intervals produced by MNR 
 for the set of parameters: $(\beta_1,\beta_2)$, $(\beta_3,\beta_4,\beta_5)$, $(\beta_1,\beta_6)$,
$(\beta_7, \beta_{10})$, and $(\beta_{20},\beta_{200},\beta_{400})$, 
 where the number in the parentheses represents the standard deviation 
 of the joint coverage rate averaged over 100 independent datasets. }
 \vspace{-0.2in}
\label{jointtab}
\begin{center}
\begin{tabular}{cccccc} \toprule
 Parameters      &  $(\beta_1,\beta_2)$   & $(\beta_3,\beta_4,\beta_5)$    & $(\beta_1,\beta_6)$   
           & $(\beta_7,\beta_{10})$    & $(\beta_{20},\beta_{200}, \beta_{400})$  \\ \midrule
 Joint coverage rate &  0.97(0.017) & 0.95(0.022) &  0.93(0.026) &  0.97(0.017) & 0.93(0.026) \\ \bottomrule
\end{tabular}
\end{center}
\end{table}

 \subsection{Simulation Studies with More Regression Models}

 \subsubsection{Linear Regression} \label{LRSect}
 
 We simulated 100 independent datasets from
 the linear regression (\ref{modeleq1}) where $n=200$, $p=500$, $\sigma^2=1$, the features were generated from a zero-mean Gaussian distribution with the precision matrix 
 $\Sigma^{-1}=\Theta=(\theta_{ij})$ given by
 \begin{equation}\label{plugin}
   \theta_{ij}=\left\{\begin{array}{ll}
                      0.5,&\textrm{if $\left| j-i \right|=1, i=2,...,(p-1),$}\\
                      0.25,&\textrm{if $\left| j-i \right|=2, i=3,...,(p-2),$}\\
                      1,&\textrm{if $i=j, i=1,...,p,$}\\
                      0,&\textrm{otherwise,}
              \end{array}\right.
\end{equation}
 and the regression coefficients were given by 
 $(\beta_0,\beta_1,\beta_2,\ldots,\beta_5)=(1,2,2.5,3,3.5,4)$ and
 $\beta_6=\cdots=\beta_p=0$.  
 Since the precision matrix has an autoregressive (AR) structure, 
 for convenience, we call this model an AR(2)-precision linear regression model. 
 
 Algorithm \ref{subsetAlg1} was first applied to this example with the numerical 
  results summarized in 
  Table \ref{regtab}, where the $\psi$-learning  algorithm was employed for Markov blanket estimation, and SIS-MCP was employed for variable selection. The $\psi$-learning algorithm has been implemented in the R-package  {\it equSA} \citep{equSApackage}. 
 It provides an equivalent measure of the partial correlation coefficient, 
 the so-called $\psi$-partial correlation coefficient, for estimating Gaussian graphical models under 
 the small-$n$-large-$p$ scenario.
 The algorithm consists of two screening stages. 
 The first stage is correlation screening,  which,   
 via a multiple hypothesis test for correlation coefficients, determines for each feature 
  a conditioning set used for calculating the $\psi$-partial correlation coefficient. 
  The second stage is $\psi$-partial correlation screening, which, via a 
  multiple hypothesis test for $\psi$-partial correlation coefficients, determines 
  the Gaussian graphical model. Corresponding to the two stages, 
  the algorithm consists of two tuning parameters, $\alpha_1$ and $\alpha_2$, which specify 
  the significance levels of the two multiple hypothesis tests, respectively.
  In all applications of the $\psi$-learning algorithm in this paper, we set 
  $\alpha_1=0.1$ and $\alpha_2=0.05$ as suggested by \cite{LiangSQ2015}. 
  In general, $\alpha_1$ should be slightly large to avoid potential loss of 
  important features in the conditioning set of each feature.  
 The nodewise regression algorithm has also been applied to this example for Markov blanket estimation, and the results are similar.

 \begin{table}[!t]
\caption{Coverage rates and widths of the 95\% confidence intervals produced by MNR, desparsified-Lasso, and ridge projection for the AR(2)-precision linear, logistic and Cox regression. Refer to the caption of Table \ref{reg0tab} for the notation.  }
 \vspace{-0.2in}
\label{regtab}
\begin{center}
\begin{tabular}{cccccc} \toprule
Response &    Measure     &                             & Desparsified-Lasso &  Ridge  & MNR \\ \midrule
         & &  signal                     & 0.2300(0.0421)   & 0.3340(0.0447)   & {\bf 0.9500(0.0218)} \\
 & \raisebox{1.5ex}{Coverage} & noise    & 0.9640(0.0186)  &  0.9922(0.0088)  &  {\bf 0.9503(0.0217)} \\ \cline{2-6}
\raisebox{1.5ex}{Gaussian} &       & signal &  0.2810(0.0027) & 0.4481(0.0043)   & 0.2806(0.0022) \\
& \raisebox{1.5ex}{Width}   & noise    &  0.2723(0.0024) & 0.4335(0.0036)   & 0.2814(0.0024) \\ \midrule
         & &  signal                     & 0.004(0.0063)   & 0(0)   & {\bf 0.9320(0.0252)} \\
 & \raisebox{1.5ex}{Coverage} & noise    & 0.9953(0.0068)  &  1.0(4.5e-4)  & {\bf 0.9373(0.0242)} \\ \cline{2-6}
\raisebox{1.5ex}{Binary} & & signal   &  0.6424(0.0101) & 1.0775(0.0110)   & 1.9473(0.0529) \\
& \raisebox{1.5ex}{Width}   & noise    &  0.5782(0.0081) & 1.0100(0.0095)   & 0.9799(0.0132) \\ \midrule
         & &  signal                     & ---  & ---   & {\bf 0.9140(0.0281)} \\
 & \raisebox{1.5ex}{Coverage} & noise    & ---  & ---   & {\bf 0.9354(0.0246)} \\ \cline{2-6}
\raisebox{1.5ex}{Survival} & & signal   &  ---  & ---   & 0.3356(0.0018) \\
& \raisebox{1.5ex}{Width}   & noise    &  ---  & ---   & 0.2683(0.0017) \\ \bottomrule
\end{tabular}
\end{center}
\end{table}

 For comparison, the desparsified Lasso and ridge projection  methods were applied to this example. 
 Both methods have been implemented in the $R$ package {\it hdi} \citep{Meierhdi2016}.
 The multi-split method is also available in {\it hdi}, but it often suffered 
 from a convergence issue in applications to this example and thus not included for comparison. 
 Table \ref{regtab} shows that MNR significantly outperforms the existing methods: 
 The coverage rates produced by MNR are almost identical to their 
  nominal levels for both zero and non-zero regression coefficients; while
 the coverage rates produced by the other methods are far from their 
 nominal levels, especially for the nonzero regression coefficients.

 For the non-zero regression coefficients, the 
 confidence intervals produced by desparsified-Lasso have about the same widths
 as those by MNR, but the former have much lower coverage rates. 
 The coverage deficiency of desparsified-Lasso 
 are due to at least two reasons: (i) the bias-corrected estimator $\hat{\bbeta}_{bc}$ 
 is still biased; and (ii) the required sparsity condition is violated. 
 The bias of $\hat{\bbeta}_{bc}$ can be easily seen from the derivation procedure of 
 $\hat{\bbeta}_{bc}$, which is due to \cite{ZhangZhang2014}.
 Let $Z_j$ denote the residual of the regression $X_j$ versus all other features 
 $\bX[-j]$, and let $P_{jk}=X_k^T Z_j/X_j^T Z_j$. Then
 the following identity holds 
  \begin{equation} \label{identeq}
  \frac{Y'Z_j}{X_j^T Z_j}=\beta_j+\sum_{k\ne j} P_{jk} \beta_k+\frac{\epsilon'Z_j}{X_j^T Z_j},
  \end{equation}
  where $Y$ and $\epsilon$ are as defined in (\ref{modeleq1}). Plugging the Lasso estimator $\hat{\bbeta}_{Lasso}$ (of the 
  regression $Y$ versus $\bX$) into (\ref{identeq}) leads to the bias-corrected estimator
  \begin{equation} \label{bceq1}
  \hat{\beta}_{bc,j}=\frac{Y'Z_j}{X_j^T Z_j} - \sum_{k \ne j} P_{jk} \hat{\beta}_{Lasso,k}
  = \hat{\beta}_{Lasso,j}+Z_j' (Y-\bX \hat{\bbeta}_{Lasso})/Z_j' X_j, \quad j=1,2,\ldots,p,
  \end{equation}
  which is essentially the same with the estimator given in (\ref{dlassoeq0}). Here 
  $\hat{\beta}_{bc,j}$ and $\hat{\beta}_{Lasso,j}$ denote the $j$-th component of
  $\hat{\bbeta}_{bc}$ and $\hat{\bbeta}_{Lasso}$, respectively. 
  The $\hat{\bbeta}_{bc}$ can have the bias of $\hat{\bbeta}_{Lasso}$ much corrected. However, 
   as implied by (\ref{identeq}), $\hat{\bbeta}_{bc}$ 
   is still generally biased because the Lasso estimator
  $\hat{\bbeta}_{Lasso}$ is generally biased.  
  Such a biased estimator shifts the center of the confidence interval and 
  thus leads to the coverage deficiency problem. 
  For the error term $\Delta_n$ defined in (\ref{dlassoeq}), 
  \cite{DezeureBMM2015} proved that it is negligible 
  if the sparsity condition $|\bS_*|=o(\sqrt{n}/\log(p))$ holds, 
  the precision matrix is row-sparse at a level of $o(n/\log(p))$, 
  and some other regularity conditions on the design matrix hold. Among these conditions, 
  the model sparsity condition $|\bS_*|=o(\sqrt{n}/\log(p))$ is a little restrictive and 
   can be easily violated.
  For example, for a problem with $|\bS_*|=5$ and $p=500$, 
  the sample size $n$ should be at least a few thousands to
  satisfy the condition $s_0 \ll \sqrt{n}/\log(p)$.
  As the result, the error term $\Delta_n$ might not be negligible, which can also 
  cause the coverage deficiency issue.  
  Since $\|\hat{\bbeta}_{Lasso}-\bbeta\|_1=O_p(|\bS_*| \sqrt{\log(p)/n})$ \citep{DezeureBMM2015}, violation of the sparsity condition also worsens the bias of $\hat{\bbeta}_{bc}$. 
  We note that the model  sparsity condition  $|\bS_*|=o(\sqrt{n})$ required by MNR  
  is much weaker than $|\bS_*|=o(\sqrt{n}/\log(p))$ under the small-$n$-large-$p$ scenario.  
    
  In our numerical experience, the coverage deficiency of desparsified-Lasso is mainly 
  due to the bias of $\hat{\bbeta}_{bc}$. We illustrate this issue using two 
  simulation studies. The first one is given as follows and the other one is given 
  in Section \ref{equicSect}. In Table \ref{biastable}, we reported 
  the values of $\hat{\bbeta}_{bc, j}$'s, $j=1,2,\ldots,8$, for the
  AR(2)-precision linear regression.  It is easy to see that
  the desparsified-Lasso estimate is severely biased for the nonzero coefficients 
  $\beta_1, \ldots, \beta_5$, which significantly shifts the centers of 
  the resulting confidence intervals and thus leads to the coverage deficiency problem. 
  Note that $\hat{\beta}_{bc,6}$ is also biased 
  due to the strong correlation between $X_6$ and $X_5$. 
  For comparison, we included in Table \ref{biastable} the MNR estimates of  
  these coefficients, which are unbiased for both zero and nonzero coefficients. 
 
\begin{table}[htbp]
\caption{Regression coefficient estimates (averaged over 100 independent datasets) 
produced by MNR and desparsified-Lasso for the AR(2)-precision linear regression 
(with $|\bS_*|=5$, $p=500$ and $n=200$), where the numbers in the parentheses represent 
the standard deviations of the estimates. }
\label{biastable}
\vspace{-0.2in}
\begin{center}
\begin{tabular}{cccccccccc} \toprule
 Method & Measure & $\beta_1$ & $\beta_2$ & $\beta_3$ & $\beta_4$ & $\beta_5$ & $\beta_6$ & $\beta_{7}$ & $\beta_{8}$ \\ \midrule
 --- & true     &  2 & 2.5 & 3 & 3.5 & 4 & 0 & 0 & 0\\ \midrule
   & $\hat{\bbeta}_{bc}$ &  1.841 & 2.274 & 2.698  & 3.270 & 3.849 & -0.051 & -0.007 & 0.016 \\ 
 \raisebox{1.5ex}{desparsified}  &  SD & (0.008) & (0.009) & (0.009) & (0.007) & (0.007)  & (0.006) & (0.007) & (0.007) 
  \\ \midrule
  & $\hat{\bbeta}_{MNR}$ &   1.997 & 2.503 & 2.994 & 3.498 & 4.001 & 0.014 & 0.004  & -0.002 \\
 \raisebox{1.5ex}{MNR}  & SD  &  (0.006) & (0.008) & (0.008) & (0.007) & (0.006) & (0.007) &  (0.008)  & (0.008) \\ \bottomrule
 \end{tabular}
 \end{center}
 \end{table}

 \subsubsection{Logistic Regression}

 We simulated 100 datasets for a logistic regression. 
 For each dataset, we set $n=300$, $p=500$, $(\beta_0,\beta_1,\ldots,\beta_5)=(1,2,2.5,3,3.5,4)$, $\beta_6=\cdots=\beta_p=0$, and generated the covariates from a 
 zero-mean multivariate Gaussian distribution with the precision matrix given by (\ref{plugin}).
 For convenience, we call this model
 an AR(2) precision logistic regression model. Each dataset consisted of 150 case samples and 150 control samples. To alleviate the convergence issues suffered by the GLM estimation procedure 
 {\it glm} in R, we set $n$ slightly large for this example. 
 
 Algorithm \ref{subsetAlg1} was run for the datasets, where the SIS-MCP algorithm was employed
 for variable selection and the $\psi$-learning algorithm was employed for 
 Markov blanket estimation. The nodewise regression algorithm was also applied for 
 Markov blanket estimation, the results were similar. 
 The numerical results were summarized in Table \ref{regtab}, which  
 indicates that MNR significantly outperforms the other methods. 
 Desparsified-Lasso and ridge projection essentially fail for this example.

 \subsubsection{Cox Regression}
 
 For Cox regression, which is also known as Cox proportional-hazards model, 
 we let $\lambda(t)$ denote the hazard rate at time $t$ and let 
 $\lambda_0(t)$ denote the baseline hazard rate. The Cox model can then be expressed as 
\begin{equation} \label{Coxeq}
\lambda(t)=\lambda_0(t) \exp(\beta_1X_1+\beta_2X_2+\ldots+\beta_pX_p). 
\end{equation}
In the simulation, we set $(\beta_1,\ldots,\beta_5)=(1,1,1,1,1)$,  $\beta_6=\cdots=\beta_p=0$,
 the baseline hazard rate $\lambda_0(t)=\lambda_0=0.1$, and the censoring hazard rate $\lambda_c=1$; 
 generated the event time from the Weibull distribution with the shape parameter=1 and the 
 scale parameter=$\lambda_0\exp(-\sum_{i=1}^p X_i \beta_i)$; generated the censoring time 
 from the Weibull distribution with the shape parameter=1 and the scale parameter=$\lambda_c$;
 set the observed survival time as the minimum of the event time and the censoring 
 time for each subject; and generated the features $X_1,\ldots, X_p$ from a zero-mean multivariate 
 normal distribution with the precision matrix given by (\ref{plugin}). For convenience,
 we call this model an AR(2)-precision Cox regression model. We simulated 100 datasets 
 from this model with $n=300$ and $p=500$. 
 
 Algorithm \ref{subsetAlg1} was run for the datasets, where the SIS-Lasso algorithm was used for 
 variable selection and the $\psi$-learning algorithm was used for  
 Markov blanket estimation. The numerical results were summarized in Table \ref{regtab}.
 The nodewise regression algorithm was also applied for 
 Markov blanket estimation, the results were similar. 
 In the MNR results, we can observe some slight bias, which mainly comes
 from the model selection error and the estimation error of the Markov blankets.
 Our numerical experience shows that the nominal level can
 be reached by MNR with the correct model and correct Markov neighborhoods
 or when $n$ becomes large.

 In addition to coverage rates, Table \ref{regtab} reports  
 mean widths of the confidence intervals resulted from different methods.
 For linear regression, the confidence intervals by MNR 
 are narrower than those by ridge projection, and of about the same width as  
 those by desparsified-Lasso.  However, as analyzed previously, 
 desparsified-Lasso often suffers from the coverage deficiency issue. 
 For logistic and Cox regression, the comparison is not meaningful, as the other methods 
 either fail or are not available.  
 
 To explore the potential of MNR in variable selection, 
 we have calculated z-scores in (\ref{zscoreq}) 
 based on the $p$-values generated by MNR for the datasets simulated above.  
 Figures S2-S4 show the histograms of the z-scores,  which indicate that the true and false features can always be well 
 separated by z-scores for all these datasets. This is an attractive feature of MNR and its use 
 for feature selection will be further explored in Section \ref{sect4}. 
 
 \subsection{Robustness of Markov Neighborhood Regression}  \label{equicSect}
 
 This section studies the robustness of MNR to violations of 
  the conditions (\ref{ceq1})-(\ref{ceq3}). 
  This issue has been partially studied in Section 2.2 of the Supplementary Material, 
  where the condition (\ref{ceq2}) is violated when the size of $\hat{\bS}_*$ is restricted to 3. 
  Recall that for the Toeplitz-covariance regression, we have 
  $|\bS_*|=5$, $|\xi_j|=2$ for $j=2,3,\ldots, p-1$, and $|\xi_j|=1$ for $j=1$ and $p$.  
  Therefore, setting $|\hat{\bS}_*| = 3$ leads to some true features missed in
  each subset regression. As shown in Table  S1 (in the Supplementary Material), this results in wider confidence intervals for both zero and nonzero 
   regression coefficients,  although the coverage rates are not much affected.
   
  In what follows, we consider one linear regression example where all features are equally
  correlated with a correlation coefficient of 0.8. The features were generated from 
  a zero-mean multivariate  Gaussian distribution with the covariance matrix given by 
  \begin{equation} \label{equieq}
   \Sigma_{i,j}=0.8, \quad \mbox{for all $i \ne j$}, \quad \Sigma_{i,i}=1 \quad \mbox{for all $i$.} 
  \end{equation}
   We set $p=500$, $n=300$, 
  $(\beta_0,\beta_1,\ldots,\beta_{10})=(1,2,2.5,3,3.5,4,$ $5,6,7,-8,-9)$,
  and $\beta_{11}= \cdots= \beta_p=0$, and generated 100 independent datasets in total. 
  For convenience, we will call this model an equi-correlation linear regression model. 
  The same model has been used in \cite{vandeGeer2014} to illustrate the performance of 
   desparsified-Lasso, but with different sample sizes and regression coefficients.
  For this example, it is easy to see that for each feature $x_j$, the Markov blanket $\xi_j$ 
  consists of all other $p-1$ features. That is, 
   the condition (\ref{ceq1}) is violated, as we always restrict the Markov blanket 
   to be much smaller  than $p$. 
  
  Algorithm \ref{subsetAlg1} was first applied to this example, where SIS-MCP was used for 
  variable selection, and nodewise regression 
   (with SIS-Lasso performed for each 
  node) was used for Markov blanket estimation. All the algorithms were run under their 
  default setting in the R package {\it SIS}.  For comparison,
  desparsified-Lasso and ridge projection methods were also applied to this example. Both 
  methods were run under their default settings in the R package {\it hdi}.  
  The numerical results were summarized in Table \ref{equictab}, which indicates that 
  MNR is pretty robust to the misspecification of the Markov blanket for this example. 
  In terms of mean coverage rates and widths, MNR produced most accurate 
  confidence intervals compared to the desparsified-Lasso and ridge projection methods. 
  
\begin{table}[htbp] 
\caption{Coverage rates and widths of the 95\% confidence intervals 
produced by desparsified-Lasso, ridge projection and MNR (with Algorithm \ref{subsetAlg1}) for
the equi-correlation linear regression model. Refer to the caption of Table \ref{reg0tab} 
for the notation. }
\label{equictab}
\begin{center}
\vspace{-0.2in}
\begin{tabular}{ccccc} \toprule
  Measure  &    ---    &  Desparsified-Lasso  &  Ridge    & MNR   \\ \hline
  & signal   &   0.916(0.028) & 0.973(0.016) & 0.938(0.024) \\
 \raisebox{1.5ex}{Coverage} & noise  &  0.963(0.019) & 0.990(0.010) &  0.951(0.022)  \\ \midrule
  & signal  &   0.656(0.003) & 1.066(0.007) &  0.551(0.003)  \\
\raisebox{1.5ex}{Width}  & noise  &  0.657(0.004) & 1.069(0.007)  &  0.554(0.003)  \\ \bottomrule
\end{tabular}
\end{center}
\vspace{-0.15in}
\end{table}
 
 Compared to the results reported in Table \ref{regtab} for the AR(2)-precision linear regression, desparsified-Lasso works much better for this example. For the  AR(2)-precision linear regression,
 Table \ref{biastable} shows that $\hat{\bbeta}_{bc}$ is 
 severely biased and thus the method suffers from coverage deficiency for nonzero coefficients. 
 To have this issue further explored, we reported 
 in Table \ref{biastable2} $\hat{\bbeta}_{bc}$ and $\hat{\bbeta}_{\rm MNR}$ for 
 the nonzero coefficients $\beta_1,\beta_2,\ldots,\beta_{10}$. 
 The comparison with the true value shows that $\hat{\bbeta}_{bc}$ is nearly unbiased for 
 $\beta_1,\ldots,\beta_8$, although it is systematically smaller than the true value 
 in magnitudes. As the result, desparsified-Lasso produced a good coverage rate for the 
 non-zero coefficients of this example. 
 MNR continuously works well; $\bbeta_{\rm MNR}$ is unbiased and accurate for this example. 
 
\begin{table}[htbp]
\caption{Regression coefficient estimates (averaged over 100 independent datasets) 
produced by MNR and desparsified-Lasso for the equi-correlation linear regression model
(with $|\bS_*|=10$, $p=500$ and $n=300$), where the numbers in the parentheses represent 
the standard deviations of the estimates. }
\label{biastable2}
\begin{center}
\vspace{-0.2in}
\begin{tabular}{cccccccccccc} \toprule
 Method &  & $\beta_1$ & $\beta_2$ & $\beta_3$ & $\beta_4$ & $\beta_5$ & $\beta_6$ & 
  $\beta_{7}$ & $\beta_{8}$ & $\beta_9$ & $\beta_{10}$ \\ \midrule
 --- & true     &  2 & 2.5 & 3 & 3.5 & 4 & 5 & 6 & 7 & -8 & -9 \\ \midrule
   & $\hat{\bbeta}_{bc}$ & 1.96 & 2.43 & 2.98 & 3.47 & 3.96 & 4.97 & 5.95 & 6.95 & -7.87 & -8.84  \\ 
 \raisebox{1.5ex}{desparsified}  &  SD &  (0.02) & (0.02) & (0.02) & (0.02) & (0.02) &  (0.02) & 
   (0.02) &  (0.02) &  (0.02)  & (0.02) \\ \midrule
  & $\hat{\bbeta}_{\rm MNR}$ &  2.01 & 2.49 & 3.00 & 3.50 & 3.99 & 5.02 & 5.98 & 6.99 & -8.01 & -9.01 \\
 \raisebox{1.5ex}{MNR}  & SD  &  (0.01) & (0.01) & (0.02) & (0.01) & (0.02) &  (0.02) & (0.02) & 
  (0.01) &  (0.02) &  (0.02) \\ \bottomrule
 \end{tabular}
 \end{center}
 \end{table}
 
 In summary, MNR is robust to misspecification of the Markov neighborhood. 
 It will perform reasonably 
 well as long as for each subset regression, the Markov neighborhood has covered the major 
 contributors of the subset regression, which include 
 the most significant features to the original  regression 
 as well as the most correlated features to the target feature of the subset regression.

\subsection{Computational Complexity} 

The MNR method consists of three steps, 
variable selection, Markov blanket estimation, and subset regression. Its computational 
complexity is typically dominated by the algorithm used for Markov blanket estimation. 
  
For Algorithm \ref{subsetAlg1},  
if the Lasso algorithm is employed for variable selection, then, by \cite{Meinshausen2007}, the computational complexity of this step is upper bounded by $O(n^3p)$ under the small-$n$-large-$p$ scenario.  Instead of Lasso, this paper employed the SCAD and MCP algorithms for variable selection which have competitive computational complexity with Lasso \citep{Zhang2010}.
If the $\psi$-learning algorithm is used for Markov blanket estimation, then, 
by \cite{LiangSQ2015}, the computational complexity of this step is upper bounded 
by $O(n^3 p^2)$. 
By condition (\ref{ceq3}), the computational complexity of each subset regression is $O(n^2)$ and thus the total computational complexity of the subset regression step is $O(n^2 p)$.
  Therefore, the total computational complexity of MNR is upper bounded by $O(n^3 p^2)$. 
Alternatively, if the nodewise regression algorithm is used for Markov blanket estimation and Lasso is used for the regression on each node/feature, then the computational complexity of
this step is also upper bounded by $O(n^3p^2)$ as there are $p$ features in total. 
In this case, the total computational complexity of MNR is also upper bounded by $O(n^3p^2)$.
If the graphical Lasso is used for Markov blanket estimation, then the total computational complexity of MNR will be $O(p^3)$, as the graphical Lasso has a computational complexity of $O(p^3)$. In a fast implementation of the graphical Lasso algorithm by making use of the block diagonal structure in its solution \citep{WittenFS2011},  the total computational complexity of MNR can be reduced to  $O(p^{2+v})$ for some $0< v \leq 1$.

 
Since desparsified-Lasso employs the nodewise regression algorithm in estimating 
the precision matrix and correcting the bias of $\hat{\bbeta}_{Lasso}$, its 
computational complexity is upper bounded by $O(n^3p^2)$, the same bound as Algorithm \ref{subsetAlg1}.

For a dataset generated from the AR(2)-precision linear regression with $p=500$ and $n=200$,  
Table \ref{timetab} summarized the CPU time cost by different methods when running with  
a single thread on an Intel(R) Xeon(R) CPU E5-2660 v3@2.60GHz machine. 
For MNR, we employed SIS-MCP for variable selection, but different methods for Markov blanket estimation. 
We note that ridge projection and multi-split can be substantially faster than MNR, although they are often inferior to MNR in numerical performance. 
For this example, ridge projection and multi-split took about 3.2 and 3.1 CPU seconds, respectively.

\begin{table}[htbp]
\begin{center}
\caption{CPU times (in seconds) cost by different methods for 
a dataset generated from the AR(2)-precision linear regression with $p=500$ and $n=200$, 
 where MNR$_a$, MNR$_b$, MNR$_c$, and MNR$_d$
 mean that $\psi$-learning, nodewise regression (with SIS-Lasso),
 nodewise regression (with SIS-MCP), and nodewise regression (with SIS-SCAD) were used for Markov blanket estimation, respectively.}
\label{timetab}
\begin{tabular}{cccccc} \toprule
Methods  &  Desparsified-Lasso &    MNR$_a$ & MNR$_b$ & MNR$_c$ & MNR$_d$ \\ \midrule
CPU(s)   &   258               &    152     &  230    & 205  &  250  \\ \bottomrule
\end{tabular}
\end{center}
\end{table}
 
Finally, we note that MNR can be substantially accelerated via parallel computing, 
for which both the Markov blanket estimation and subset regression steps 
can be done in an embarrassingly parallel way.
As described in Section \ref{LRSect}, the $\psi$-learning algorithm consists of two 
screening stages, for which both the correlation coefficients and $\psi$-partial 
correlation coefficients can be calculated in parallel. Refer to \cite{LiangSQ2015} for 
more discussions on this issue. If the nodewise regression algorithm is used for 
Markov blanket estimation, its parallel implementation is obvious. 

\section{Causal Structure Discovery for High-Dimensional Regression}  \label{sect4}

 The causal relationship for a pair or more variables refers to a {\it persistent association}
 which is expected to exist in all situations without being affected by
 the values of other variables. Due to its attractive feature, which does not only allow
 better explanations for past events but also enables better predictions for the future,
 causal discovery has been an essential task in many disciplines.
 Since, for high-dimensional problems, it is difficult and expensive to identify causal relationships
 through intervention experiments, passively observed data has thus become an important source to
 be searched for causal relationships. The challenge of causal discovery from observational
 data lies in the fact that statistical associations detected from observational data
 are not necessarily causal. 

 In statistics, the causal relationship or {\it persistent association} can be determined using
 conditional independence tests. For a large set of variables, a pair of variables
 are considered to have no direct causal relationship if a subset of the remaining variables
 can be found such that conditioning on this subset of variables, the two variables are
 independent. Based on conditional independence tests,
 \cite{SpirtesC2000} proposed the famous PC algorithm
 for learning the structure of causal Bayesian networks.
 Later, \cite{BuhlmannMM2010} extended
 the PC algorithm to high-dimensional variable selection. The extension is called
  the PC-simple algorithm which can be used to search for the causal
 structure around the response variable. Note that the causal structure
 includes all the possible direct causes and effects of the response variable, i.e., all
 the parents and children in the terminology of directed acyclic graphs (DAGs).
 For certain problems, we may be able to determine in logic which are for parents and
 which are for children, although PC-simple cannot tell.
 An alternative algorithm that can be used for local causal discovery is
 the HITON-PC algorithm \citep{Aliferisetal2010}, which is also an extension of the PC algorithm.
 The major issue with the PC-simple and HITON-PC algorithms is with their time complexity.
 For both algorithms, in the worst scenario, i.e., when for each of the $p$ features all conditional
 independence tests of order from 1 to $p-1$ are conducted, the total number of
 conditional tests is $O(p 2^p)$.
 Even under the sparsity constraint,
 the total number of conditional tests can still be of a high order polynomial of $p$.
 See \cite{BuhlmannMM2010} for more discussions on this issue.

 In what follows we describe how the causal structure
 around the response variable can be discovered for high-dimensional regression based on the
 output of MNR.  The proposed algorithm has a much favorable
 computational complexity, which is $O(n^{3/2}p)$ in all scenarios.  
 For Gaussian, binary and proportional-hazards response data, 
 the MNR method can be described one by one as follows. 

 \subsection{Gaussian Response} \label{CausalGaussoian}

  Assume that $\bZ=(Y,\bX)$ jointly follows a multivariate Gaussian distribution $N_{p+1}(0,\Sigma)$.
  To distinguish the notation from that used in previous sections, we let 
  $\bG_z=(\bV_z,\bE_z)$ denote the graph underlying the joint Gaussian distribution.
  Let $\zeta_j=\xi_j \cup \bS_*$, where $\xi_j$ denotes the
  Markov blanket of $X_j$ in the graph $\bG_z$.
  It is easy to see that $\zeta_j$ forms
  a separator of $Y$ and $X_j$ in the graph $\bG_z$.
  Then, under the faithfulness condition for the joint distribution $N_{p+1}(0,\Sigma)$, we can show
  as in \cite{LiangSQ2015} that
 $ \rho(Y, X_j|\bX_{\bV\setminus \{j\}}) \ne 0 \Longleftrightarrow \rho(Y, X_j | \bX_{\zeta_j}) \ne 0$,
  where $\rho(\cdot,\cdot|\cdot)$ denotes the partial correlation coefficient.
  The validity of the faithfulness condition is supported by the Lebesgue
  measure zero argument \citep{Meek1995}; 
 that is, the problems that violate the faithfulness condition usually correspond to some
 particular parameter values that form a zero measure set in the space of all possible parameterizations.
  Further, by the relationship between partial correlations
  and regression coefficients, see e.g., p.436 of \cite{BuhlmannGeerBook2011}, we have
   $\rho(Y, X_j | \bX_{\zeta_j}) \ne 0 
   \Longleftrightarrow \beta_j \ne 0$, where $\beta_j$ is the
   coefficient of $X_j$ in the Markov neighborhood regression $Y \sim X_j+\bX_{\zeta_j}$.
   Therefore, the test for $H_0: \beta_j=0$ versus $H_1: \beta_j \ne 0$
   can be conducted via the Markov neighborhood regression.
   Given the $p$-values of individual tests,
   the causal structure around the response variable $Y$ can be determined
   via a multiple hypothesis test.  

  Since the problem of causal structure discovery is to identify
  a small set of variables that have causal or effect relations with the response variable,
  a simplified version of Algorithm S1 can be used, which avoids to assess the
  effect of all variables on the response.  In the simplified algorithm,
  the Markov blankets only need to be found for the variables
  survived from the variable screening step. The simplified algorithm can be
  described as follows.

\begin{algorithm} (Simplified MNR for Causal Structure Discovery) \label{subsetAlg3}
\begin{itemize}
   \item[(a)] (Variable screening)
    Apply a sure independence screening procedure with $Y$ as
    the response variable and $\bX$ as features,
    to obtain a reduced feature set, $\hat{\bS}_* \subseteq \{1,\dots,p\}$, with
    the size $|\hat{\bS}_*|=O(\sqrt{n}/\log(n))$.

   \item[(b)] (Markov blanket estimation)
    For each variable $X_j \in \hat{\bS}_*$, apply a sure independence screening
    procedure to obtain a reduced neighborhood
   $\hat{\xi}_{j}\subseteq \{1,\dots,p\}$ with the size $|\hat{\xi}_j|=O(\sqrt{n}/\log(n))$.

  \item[(c)] (Subset Regression) For each feature $X_j\in \hat{\bS}_*$, run
  a subset regression with the features given by $\{X_j\} \cup \bX_{\hat{\xi}_j} \cup \bX_{\hat{\bS_*}}$.
  Conduct inference for $\beta_j$, including the estimate, confidence interval and $p$-value,
  based on the output of the subset regression.

 \item[(d)] (Causal Structure Discovery) Conduct a multiple hypothesis test to identify causal features
          based on the $p$-values calculated in step (c).
\end{itemize}
\end{algorithm}

  The consistency of  the algorithm for causal structure identification can be established 
  under slightly modified conditions of Theorem \ref{Them1}. To be more precise, 
  we only need to restate the conditions (A1)-(A4)  for the joint distribution of $(Y,\bX)$,
  and then the proof directly follows Theorem 2 of \cite{LiangSQ2015}.
 It is easy to see that the computational complexity of this algorithm is $O(n^{3/2}p)$, as the 
 computational complexity of the SIS algorithm is $O(np)$ \citep{FanLv2008} and there are 
 a total of $O(\sqrt{n}/\log(n))$ features for which the Markov blanket needs to be estimated. 
 Hence, Algorithm \ref{subsetAlg3} can potentially be much faster than the PC-simple 
 and HITON-PC algorithms, especially when $p$ is large. 
 Again, this algorithm can have many implementations. For example, 
 the SIS algorithm \citep{FanLv2008} can be used for both variable screening 
 and Markov blanket estimation. 
 The HZ-SIS algorithm \citep{XueLiang2017} can also be used for both of them.

  \subsection{Binary Response}

   For binary response data, if we assume that the features $\bX$ follow a Gaussian
   distribution, then a joint distribution of $(Y,\bX)$ can be defined as in \cite{LeeHastie2015},
    for which the conditional distribution of each component of $\bX$
    is Gaussian with a linear regression model, and
    the conditional distribution of $Y$ is a binomial distribution as given
    by a logistic distribution.
    Further, we can assume that the joint distribution is faithful to
    the mixed Graphical model formed  by $(Y,\bX)$ \citep{Meek1995}.
    As pointed out in \cite{LeeHastie2015}, the mixed graphical model is a
   pairwise Markov network and the zero regression coefficients (in the
   nodewise regression) correspond to the conditional independence.
   Therefore, Algorithm \ref{subsetAlg3} is also applicable to the binary response data,
   for which variable screening can be done using the GLM SIS algorithm \citep{FanS2010}.
   Extending the algorithm to multinomial response data is straightforward. The consistency 
   of the approach directly follows from Theorem 2 of \cite{XuJiaLiang2019}, which shows
    the consistency of a conditional independence test based approach for learning 
    mixed graphical models. Following \cite{XuJiaLiang2019}, the consistency of the 
    proposed approach can be established under appropriate conditions 
    such as the faithfulness of the joint distribution of $(Y,\bX)$ with respect to 
    the underlying mixed graphical model, the sparsity of Markov blankets, the sparsity 
    of the true model, and some conditions on generalized linear models. 
    
  \subsection{Proportional-Hazards} \label{Coxsection}

   For Gaussian and binary response data, we justify Algorithm \ref{subsetAlg3}
   for causal structure discovery by presenting $(Y,\bX)$ as an undirected graph for which
   the causal structure around $Y$ contains both direct causes and effects
   of the response variable. Unfortunately, extending this justification to 
   survival data is hard.  
   For survival data, the response variable is proportional hazard, which is non-Gaussian and non-multinomial and thus the joint distribution of $(Y,X)$ 
    is difficult to define with respect to an undirected graph.
   However, this difficulty can be resolved by modeling $(Y,X)$ as 
   a Bayesian network with $Y$ being a child node only. 
   If $Y$ is a child of $X_j$, then
   $\tilde{\zeta}_j= \xi_j \cup \{p+1\} \cup \bS_*$ forms the Markov blanket
   of $X_j$, where $\xi_j$ is the sub-Markov blanket formed with $\bX$
   as implied by the PC algorithm \citep{SpirtesC2000}, $p+1$ is the index of $Y$ (by defining $X_{p+1}=Y$),
   and $\bS_*$ contains all siblings of $X_j$ with respect to the common child $Y$.
   By the total conditioning property shown in \cite{PelletE2008} for Bayesian networks, we have
   \begin{equation} \label{BNeq1}
    X_j \perp Y |\bX_{\tilde{\zeta}_j\setminus \{j,p+1\}} \Longleftrightarrow 
    X_j \perp Y |\bX_{\bV \setminus \{j,p+1\}},
   \end{equation}
    which implies that Algorithm \ref{subsetAlg3} is still valid for survival data.
    However, construction of Bayesian networks for non-Gaussian and non-multinomial and 
    with missing data is beyond the scope of this paper. Therefore, there will be no 
    illustrative examples for this part. 
    In (\ref{BNeq1}), if $\tilde{\zeta}_j$ is replaced by some super Markov blanket
    $\tilde{\zeta}_j' \supset \tilde{\zeta}_j$, the equivalence still holds.

   This justification is very general and can be applied to the Gaussian and
   multinomial response data as well. The only shortcoming is that it assumes that $Y$ can
   only be a child of $\bX$, while this might be too restrictive for the problems considered
   with the Gaussian and multinomial response data.

\section{Real Data Studies}

 This section reports two applications of Algorithm \ref{subsetAlg3},
 one is for identification of anti-cancer drug sensitive genes, and the other is for
 identification of cancer driver genes. 

 \subsection{Identification of Drug Sensitive Genes} 

 Disease heterogeneity is often observed in complex diseases such as cancer.
 For example, molecularly targeted cancer drugs are only
 effective for patients with tumors expressing targets \citep{GrunwaldH2003, Buzdar2009}.
 The disease heterogeneity has directly motivated the development of precision medicine,
 which aims to improve patient care by tailoring optimal therapies to an individual patient according to his/her
 molecular profile and clinical characteristics. 
 Identifying sensitive genes to different drugs
 is an important step toward the goal of precision medicine. 

 To illustrate the MNR method, we considered the cancer cell line encyclopedia (CCLE) dataset, 
  which  is publicly available at {\it www.broadinstitute.org/ccle}.
  The dataset consists of 8-point dose-response curves for 24 chemical compounds across over 400
 cell lines. For different chemical compounds, the numbers of cell lines are
 slightly different. For each cell line, it consists of the expression values  
 of $p=18,988$ genes.  We used the area under the dose-response curve, which was termed
 as activity area in \cite{Barretinaetal2012}, to measure the sensitivity of
 a drug to each cell line. Compared to other measurements, such as $IC_{50}$
 and $EC_{50}$, the activity area could capture the efficacy and potency of the drug
 simultaneously.  An exploratory analysis indicates that treating the activity area 
 as the response of a linear regression with respect to 
 the gene expression values is appropriate.   
 
 Since the purpose of this study is to identify the drug sensitive genes instead of 
 assessing the drug effect for all genes, Algorithm \ref{subsetAlg3} was applied 
 with the HZ-SIS algorithm used for variable screening and 
 Markov blanket estimation. In both steps, we set the neighborhood  
 size to be 40. After getting $p$-values from the subset regressions, 
 the adjusted $p$-values \citep{Holm1979} were calculated,
 and the genes with the adjusted $p$-values less than 0.05 were identified 
 as the drug sensitive genes. 
 For some drugs, if there are no genes identified 
 at this significance level, we just selected 
 one gene with the smallest $p$-value. 
 The results were summarized in Table \ref{drugtab}. 
 For Algorithm \ref{subsetAlg3}, different neighborhood sizes have been tried, the 
 results are similar. 

  For comparison, 
  desparsified Lasso, ridge projection and multi sample-splitting 
  were also applied to this example. As in the MNR method, 
  for each drug, we selected the genes with 
  the adjusted $p$-values less than 0.05 
  as significant; and if there were no genes selected at this significance level, we just reported 
  one gene with the smallest adjusted $p$-value.  
  The results were also summarized in Table \ref{drugtab}.

 Compared to the existing methods, MNR performs reasonably well for this real data example. 
 First of all,
  for all drugs, desparsified Lasso is simply inapplicable due to 
  the ultra-high dimensionality of the dataset; 
  the package {\it hdi} aborted due to the excess of memory limit.  
  Due to the same issue, {\it hdi} also aborted for some drugs when performing 
  ridge regression.  For multi sample-splitting and MNR, it is easy to see that if the same gene 
  is selected by both methods, then the 95\% confidence interval 
  produced by MNR is narrower.

  MNR produced promising results in selection of drug sensitive genes. 
  For example, for both drugs Topotecan and Irinotecan, MNR selected 
  the gene SLFN11 as the top drug sensitive gene.  
  In the literature, \cite{Barretinaetal2012} and \cite{Zoppolietal2012}
   reported that SLFN11 is predictive of treatment response for Topotecan and Irinotecan.
  For drug 17-AAG, MNR selected NQO1 as the top gene; 
  in the literature, \cite{HadleyH2014} and \cite{Barretinaetal2012}
  reported NQO1 as the top predictive biomarker for 17-AAG.
   For drug Paclitaxel, BNN selected BCL2L1 as the top gene. 
   In the literature, many publications, such as \cite{LeeHLYK2016} and
   \cite{Domanetal2016}, reported that the gene BCL2L1 is predictive of treatment response
   for Paclitaxel.
   For drug PF2341066, \cite{LawrenceSalgia2010} reported that HGF, which
 was selected by MNR as the top drug sensitive gene,
 is potentially responsible for the effect of PF2341066.
 For drug LBW242, RIPK1 is selected by MNR. \cite{Gaither2007} and \cite{Moriwaki2015} 
 stated that RIPK1 is one of the presumed target of LBW242, which is involved in increasing death of cells. 
 Finally, we pointed out that the genes selected by MNR have some overlaps with those selected by 
 the multi sample-splitting method, although for the overlapped genes the 95\% confidence intervals produced by MNR tend 
 to be narrower. 


{\small
\begin{center}
\begin{longtable}{ccccc}
\caption{ Comparison of drug sensitive genes selected by desparsified Lasso, ridge projection,
  multi sample-splitting (multi-split) and MNR for 24 anti-cancer drugs, where $^*$ indicates that 
  this gene was significantly selected and the number in the parentheses denotes the
  width of the 95\% confidence interval produced by the method.} 
\label{drugtab} \\ \toprule 
Drug &Desparsified Lasso & Ridge & Multi-Split  & MNR \\ \hline
\endfirsthead

\multicolumn{5}{c}%
{{\bfseries \tablename\ \thetable{} -- continued from previous page}} \\
\hline Drug &Desparsified Lasso & Ridge  & Multi-Split  & MNR \\ \hline
\endhead

\hline \multicolumn{5}{|r|}{{Continued on next page}} \\ \hline
\endfoot
\endlastfoot
17-AAG&\makecell{--}&\makecell{--}&\makecell{NQO1*(0.138)}&\makecell{NQO1*(0.115)}\\\hline 
AEW541&\makecell{--}&\makecell{F3(0.076)}&\makecell{SP1(0.176)}&\makecell{TMEM229B*(0.142)}\\\hline 
AZD0530&\makecell{--}&\makecell{PPY2(0.966)}&\makecell{SYN3(0.705)}&\makecell{DDAH2(0.088)}\\\hline 
AZD6244&\makecell{--}&\makecell{OSBPL3(0.161)}&\makecell{SPRY2*(0.084)\\LYZ*(0.069)\\RNF125*(0.084)}&\makecell{LYZ*(0.048)\\SPRY2*(0.056)}\\\hline 
Erlotinib&\makecell{--}&\makecell{LRRN1(0.102)}&\makecell{PCDHGC3(0.684)}&\makecell{ENPP1(0.123)}\\\hline 
Irinotecan&\makecell{--}&\makecell{SLFN11(0.091)}&\makecell{ARHGAP19*(0.134)\\SLFN11*(0.044)}&\makecell{ARHGAP19*(0.108)\\SLFN11*(0.033)}\\\hline 
L-685458&\makecell{--}&\makecell{--}&\makecell{MSL2(0.2)}&\makecell{FAM129B(0.187)}\\\hline 
Lapatinib&\makecell{--}&\makecell{WDFY4(0.509)}&\makecell{ERBB2*(0.111)}&\makecell{SYTL1(0.062)}\\\hline 
LBW242&\makecell{--}&\makecell{RXFP3(0.86)}&\makecell{LOC100009676(0)}&\makecell{RIPK1(0.221)}\\\hline 
Nilotinib&\makecell{--}&\makecell{--}&\makecell{RAB37(0.187)}&\makecell{RHOC(0.103)}\\\hline 
Nutlin-3&\makecell{--}&\makecell{TTC7B(0.119)}&\makecell{LOC100009676(0)}&\makecell{DNAJB14(0.163)}\\\hline 
Paclitaxel&\makecell{--}&\makecell{ABCB1*(0.229)}&\makecell{ABCB1*(0.183)}&\makecell{BCL2L1*(0.289)}\\\hline 
Panobinostat&\makecell{--}&\makecell{C17orf105(1.104)}&\makecell{PUM2(0.589)}&\makecell{TGFB2(0.103)}\\\hline 
PD-0325901&\makecell{--}&\makecell{ZNF646(0.498)}&\makecell{LYZ*(0.064)\\RNF125*(0.087)}&\makecell{DBN1(0.104)}\\\hline 
PD-0332991&\makecell{--}&\makecell{GRM6(0.719)}&\makecell{LOC100506972(0.569)}&\makecell{PUM2(0.244)}\\\hline 
PF2341066&\makecell{--}&\makecell{WDFY4(0.487)}&\makecell{SPN*(0.124)}&\makecell{HGF*(0.043)\\ENAH*(0.068)\\GHRLOS2*(0.24)}\\\hline 
PHA-665752&\makecell{--}&\makecell{--}&\makecell{LAIR1(0.193)}&\makecell{INHBB(0.039)}\\\hline 
PLX4720&\makecell{--}&\makecell{ADAMTS13(0.692)}&\makecell{SPRYD5*(0.118)}&\makecell{PLEKHH3(0.22)}\\\hline 
RAF265&\makecell{--}&\makecell{LOC100507235(0.748)}&\makecell{SIGLEC9(0.761)}&\makecell{SEPT11*(0.078)}\\\hline 
Sorafenib&\makecell{--}&\makecell{--}&\makecell{SBNO1(0.426)}&\makecell{RPL22*(0.151)\\LAIR1*(0.094)}\\\hline 
TAE684&\makecell{--}&\makecell{--}&\makecell{ARID3A*(0.11)}&\makecell{ARID3A*(0.078)}\\\hline 
TKI258&\makecell{--}&\makecell{--}&\makecell{SPN(0.12)}&\makecell{KHDRBS1(0.251)}\\\hline 
Topotecan&\makecell{--}&\makecell{--}&\makecell{SLFN11*(0.136)}&\makecell{SLFN11*(0.107)}\\\hline 
ZD-6474&\makecell{--}&\makecell{MID1IP1(0.158)}&\makecell{NOD1(0.363)}&\makecell{PXK*(0.066)}\\ \bottomrule
\end{longtable}
\vspace{-0.15in}
\end{center}
}

\subsection{Identification of Cancer Driver Genes} 

We considered the Lymph dataset \citep{HansDW2007}, which consists of $n = 148$ samples 
 with 100 node-negative cases (low risk for breast cancer) and 48 node-positive cases 
(high risk for breast cancer) as our binary response. 
 For each sample, there are $p=4512$ genes that showed evidence of variation above 
 the noise level for further study. This dataset has been analyzed 
 by multiple authors, such as \cite{HansDW2007} and \cite{LiangSY2013}. 

 Algorithm \ref{subsetAlg3} was applied to this dataset, where variable screening 
 was done using the GLM SIS algorithm developed in \cite{FanS2010},  
 and the Markov blanket estimation step 
 was done using the HZ-SIS algorithm \citep{XueLiang2017}.  
 In both steps, we set the neighborhood size to be 5.  For this dataset, 
 MNR selected two genes, RGS3 and ATP6V1F, with the adjusted $p$-value less than 0.05. 
 The details were given in Table \ref{cancertab}, which 
 are consistent with our existing knowledge. 
 For example, RGS3 is known to play a role in modulating the ability of motile lymphoid cells \citep{Bowman1998},
 and to be upregulated in p53-mutated breast cancer tumors \citep{Ooe2007}. 
 ATP6V1F has been reported by many authors in lymph node status studies, see e.g., 
 \cite{HansDW2007} and \cite{Dobra2009}. 
 For comparison, desparsified Lasso and ridge projection methods 
 were also applied to this example. As aforementioned, the multi sample-splitting algorithm is 
 not yet available for logistic regression. 
 Both desparsified Lasso and ridge projection selected only the gene RGS3 as the cancer driver gene. 
 
 A closer look at Table \ref{cancertab} shows that MNR outperforms 
 desparsified Lasso and ridge projection for this example. This can be explained from two perspectives. 
 First, MNR is the only method that identifies RGS3 as a cancer driver gene at an acceptable 
 significance level. While  
 the desparsified Lasso and ridge projection can only identify that RGS3 has a smaller adjusted $p$-value 
 than other genes, and its adjusted $p$-value is greater than 0.05.    
 Second, for the gene RGS3, the 95\% confidence interval produced by MNR is narrower than 
 that produced by desparsified Lasso. Moreover, the 95\% confidence interval 
 produced by ridge projection even contains 0 and is thus less significant.  

\begin{table}
\selectfont
\begin{center}
\caption{Comparison of the cancer driver genes selected by the MNR, desparsified Lasso and ridge projection methods 
 for the Lymph dataset, where $^*$ indicates that
  this gene was significantly selected. }
\vspace{-5mm}
\label{cancertab}
\begin{tabular}{ccccccc} \\ \toprule
 & Desparsified Lasso & &  Ridge & &  \multicolumn{2}{c}{MNR}  \\ \cline{2-2}  \cline{4-4} \cline{6-7}
 Gene & RGS3 & & RGS3 & & RGS3* & ATP6V1F*  \\ 
 95\%C.I. &  (1.145,5.748)& &  (-0.251,2.249)& & (0.859,5.178) & (2.073,7.131)\\ 
 Width & 4.603 &&  2.500 & & 4.319 & 5.058\\ \bottomrule
\end{tabular}
\vspace{-0.15in}
\end{center}
\end{table}

\section{Discussion} 

 This paper has proposed the MNR method 
 for constructing confidence intervals and assessing 
 $p$-values for high-dimensional regression. The MNR method 
 has successfully broken the high-dimensional inference problem into
 a series of low-dimensional inference problems based on 
 conditional independence relations among different variables.  
 The embarrassingly parallel structure 
 of the MNR method, where the Markov blanket,
 confidence interval and $p$-value can be calculated for each variable in parallel, 
 enables it potentially to be run very fast on multicore computers. 
 The MNR method has been tested on high-dimensional linear,
 logistic and Cox regression. The numerical results indicate that the MNR method
 significantly outperforms the existing ones.
 The MNR method has been applied to learn causal structures for
 high-dimensional linear models with the real data examples
 for identification of drug sensitive genes and cancer driver genes presented.

 This paper has assumed that the features 
 are Gaussian. Extension of the MNR method to non-Gaussian features  
 is straightforward. In this case, the conditional independence relations 
 among the features can be figured out using Bayesian networks 
 based on the concept of Markov blanket as described in Section \ref{Coxsection}. 
 The theory developed in Section 2 will still hold. 
 The idea of using conditional independence relations for dimension
 reduction is general and potentially can be extended to
 other high-dimensional or big data problems as well.
 
 Finally, we note that the performance of the MNR method relies on the algorithms used for
 variable selection and Markov blanket estimation, while each of these algorithms 
 can rely on a non-trivial amount of tuning. A sub-optimal performance of these 
 algorithms may adversely affect the performance of the MNR method, for example, the resulting 
 confidence intervals can be wider. 



\section*{Acknowledgments}
The authors thank the editor, associate editor and two referees for their encouraging and 
constructive comments which 
have led to significant improvement of this paper. 

\vspace{10mm}

\appendix 

\noindent {\Large \bf Appendix}

\section{Conditions for Theorem \ref{Them1}}

\begin{itemize} 
        \item[(A0)] The dimension $p_{n}=O(\exp(n^{\delta}))$ for some constant $0\leq \delta<1/2$.
        \item[(A1)] The distribution $P_{\bX}$ is multivariate Gaussian, and it  
           satisfies the Markov property and adjacency faithfulness condition 
           with respect to the undirected underlying graph $\bG$.
        \item[(A2)] The correlation satisfy $\min \{|r_{ij}|;e_{ij}=1, i,j=1,\dots,p_{n}, i\neq j \}\geq c_{0}n^{-\kappa}$
         for some constants $c_{0}>0$ and $0<\kappa<(1-\delta)/2$, and
        $\max \{|r_{ij}|; i,j=1,\dots,p_{n}, i\neq j \}\leq M_{r} < 1$ for some constants $0<M_{r}<1$.
        \item[(A3)] There exists constants $c_{1}>0$, $0<\kappa' \leq \kappa$, and $0\leq \tau < 1-2\kappa'$
              such that $\lambda_{\max}(\Sigma)\leq c_{1}n^{\tau}$.
        \item [(A4)] The $\psi$-partial correlation coefficients satisfy
        $\inf\{\psi_{ij}:\psi_{ij}\neq 0, i,j=1,\dots,p_{n}, i\neq j, |S_{ij}|\leq q_{n}\} \geq c_{2}n^{-d}$,
        where $q_{n}=O(n^{2\kappa'+\tau}), 0<c_{2}<\infty$, $0<d<(1-\delta)/2$ are some constants, 
        and $S_{ij}$ denotes the conditioning set used in calculating $\psi_{ij}$.
         In addition, $\sup\{\psi_{ij}: i,j=1,\dots,p_{n}, i\neq j, |S_{ij}|\leq q_{n}\} \geq c_{6}n^{-d}\leq M_{\psi}<1$
        for some constants $0<M_{\psi}<1$.
        \item [(A5)] $\max_{j=1,\dots,p}|\xi_{j}|=o(n^{1/2})$.
        \item [(A6)] There exist constants $c_{3}>0$  and $c_{4}>0$ such that
        $\min _{j\in \bS_*}|\beta_{j}|\geq c_{3}n^{-\kappa}$ and
        $\min_{j\in \bS_*}|cov(\beta_{j}^{-1}y,x^{(j)})|\geq c_{4}$.

        \item[(A7)] $|\bS_*|=o(n^{1/3})$.
        \item[(A8)] Other assumptions in Theorem 2 of \cite{Fan2004}.
        \item[(A9)] Other assumptions in Theorem 1 of \cite{FanGH2012} (for the case of random design).
\end{itemize}

\newpage

\section{Supplementary Materials} 

\setcounter{table}{0}
\renewcommand{\thetable}{S\arabic{table}}
\setcounter{figure}{0}
\renewcommand{\thefigure}{S\arabic{figure}}
\setcounter{equation}{0}
\renewcommand{\theequation}{S\arabic{equation}}
\setcounter{algorithm}{0}
\renewcommand{\thealgorithm}{S\arabic{algorithm}}
\setcounter{lemma}{0}
\renewcommand{\thelemma}{S\arabic{lemma}}
\setcounter{theorem}{0}
\renewcommand{\thetheorem}{S\arabic{theorem}}
\setcounter{remark}{0}
\renewcommand{\theremark}{S\arabic{remark}}

 This material is organized as follows. Section \ref{SMsect1} gives the proofs of Lemma 1, Lemma 2 and Theorem 1. Section \ref{SMsect2} presents a variable screening-based MNR method, justifies its validity, and illustrates its performance using a numerical example. 
 Section \ref{SMsect3} presents some figures which illustrate the performance of the MNR method in variable selection.

\subsection{Proofs for the Validity of Algorithm 1}
\label{SMsect1}

\subsubsection{Proof of Lemma 1}

\begin{proof}
 Without loss of generality, we let $j=1$ and $\{j\}\cup\hat{\xi}_{j}\cup \hat{\bS}_*=\{1,\dots,d\}=D_j$. 
 Let $u^{(i)}=y^{(i)}-(\bx_{D_j}^{(i)})^{T}\boldsymbol{\beta}_{D_j}^{\ast}$, then 
\[
\sqrt{n}(\hat{\bbeta}_{D_j}-\bbeta_{D_j})=\bigg[\sum_{i=1}^{n}\bx^{(i)}_{D_j}(\bx^{(i)}_{D_j})^{T}\bigg]^{-1}
 \bigg[\sqrt{n}\sum_{i=1}^{n}\bx^{(i)}_{D_j}u_{i}\bigg]
 =\bigg[\frac{1}{n} \sum_{i=1}^{n}\bx^{(i)}_{D_j}(\bx^{(i)}_{D_j})^{T}  \bigg]^{-1}
  \bigg[\frac{1}{\sqrt{n}}\sum_{i=1}^{n}\bx^{(i)}_{D_j}u_{i}\bigg].
\]
Since $\bS_*\subseteq \hat{\bS}_* \subseteq \{1,\dots,d\}$, it's easy to verify
$E(\bx^{(i)}_{D_j}u_{i})=E_{\bx^{(i)}_{D_j}}E(\bx^{(i)}_{D_j}u_{i}|\bx^{(i)}_{D_j})=\mathbf{0}$. In addition,
\[
Cov(\bx^{(i)}_{D_j}u_{i})=E_{\bx^{(i)}_{D_j}}Cov(\bx^{(i)}_{D_j}u_{i}|\bx^{(i)}_{D_j})
 +Cov_{\bx^{(i)}_{D_j}}E(\bx^{(i)}_{D_j}u_{i}|
  \bx^{(i)}_{D_j})=\sigma^{2}\Sigma_{d},
\]
where $\Sigma_{d}$ denotes the upper left $d\times d$ submatrix of the covariance matrix $\Sigma$ (of $\bX$). 

When $|D_j|=d=o(n^{1/2})$, by Theorem 1.1 of \cite{Portnoy1986}, 
\[
\frac{1}{\sqrt{n}}\sum_{i=1}^{n}\bx^{(i)}_{D_j}u_{i} \xrightarrow{d} N(0,\sigma^2 \Sigma_{d}).
\]
By the weak law of large numbers, 
$\frac{1}{n} \sum_{i=1}^{n}\bx^{(i)}_{D_j}(\bx^{(i)}_{D_j})^{T} \xrightarrow{p} \Sigma_{d}$.
Further, by continuous mapping theorem, 
\begin{equation} \label{esttheta}
\bigg[\frac{1}{n} \sum_{i=1}^{n}\bx^{(i)}_{D_j}(\bx^{(i)}_{D_j})^{T} \bigg]^{-1} \xrightarrow{p} \Sigma_{d}^{-1}.
\end{equation}
Combining them together, we have
$\sqrt{n}(\hat{\bbeta}_{D_j}-\bbeta_{D_j}) \xrightarrow{d} N(0,\sigma^{2}\Sigma_{d}^{-1})$. 
 Therefore $\sqrt{n}(\hat{\beta}_{1}-\beta_{1}) \xrightarrow{d} \sigma^{2}\theta_{d,11}$, 
 where $\theta_{d,11}$ denotes the (1,1)th entry of the matrix $\Sigma_{d}^{-1}$.   
 As shown around the equations (4) and (5) in the main text,
 $\theta_{d,11}=\theta_{11}$, which completes the proof of part (i). 

  Part (i) implies that $\sqrt{n} \frac{\hat{\beta}_j-\beta_j}{\sqrt{ {\sigma}^2 {\theta}_{jj}}} 
  \sim N(0,1)$ as $n\to \infty$.
  Since $\hat{\bS}_* \supseteq \bS_*$ and $|D_j|=o(n^{1/2})$,  
  $\hat{\sigma}_n^2$ is a consistent estimator of $\sigma^2$.
  By (\ref{esttheta}), $\hat{\theta}_{11}$ forms a consistent estimator of $\theta_{11}$.
  Putting all these together, we have
  $\sqrt{n} \frac{\hat{\beta}_j-\beta_j}{\sqrt{ \hat{\sigma}_n^2 \hat{\theta}_{jj}}} \sim N(0,1)$ 
  as $n\to \infty$.
 \end{proof}

 \subsubsection{Proof of Theorem 1} 
 
 \begin{proof} 
 With conditions (A0)-(A4), we can apply Theorem 2 of \cite{LiangSQ2015} to
 prove that $P(\hat{\xi}_{j}=\xi_{j})=1-o(1)$ for all $j=1,\dots,p_{n}$.
 Similarly, with conditions (A0), (A3) and (A6)-(A8), we can apply Theorem 5 of \cite{FanLv2008}
 to prove that $P(\hat{\bS}_*=\bS_*)=1-o(1)$.
 With conditions (A5) and (A7), we have $P(|D_j|=|\{j\}\cup\hat{\xi}_{j}\cup \hat{\bS}_*|=o(n^{1/2}))=1-o(1)$.
 Following from Lemma 1, we can derive the asymptotic distribution of $\hat{\beta}_{j}$  as expressed
 in part (i) of Lemma 1.
 By condition (A9), we can get the consistency of $\hat{\sigma}_n^2$
 based on Theorem 1 of \cite{FanGH2012} and the asymptotic $P(|D_j|\log(p_n)=o(n))=1-o(1)$.   
 Finally, the proof can be concluded based on Slutsky's theorem.  
\end{proof}

\subsubsection{Proof of Lemma 2}

 
\begin{proof}
Without loss of generality, we let $j=1$ and $\{j\}\cup\hat{\xi}_{j}\cup \hat{\bS}_*=\{1,\dots,d\}=D_j$.
Since $|D_j|=o(n^{1/2})$ and $\hat{\bS}_* \supseteq \bS_*$ hold, by the theory of 
 GLM estimation with increasing dimensions \citep{Portnoy1988}, 
 $\sqrt{n}(\hat{\beta}_{1}-\beta_{1}) \xrightarrow{d} N(0,k_{d,11})$,
 where $k_{d,11}$ denotes the $(1,1)$-th entry of the inverse of the Fisher information matrix $I_{d}^{-1}$ 
 for the subset GLM. 

 In order to show $k_{d,11}=k_{11}$, we first prove a useful proposition: 
 {\it For any two features $X_i$ and $X_j$, if at least one of them is not in $\bS_*$, 
 then $k_{ij}=0 \Leftrightarrow \theta_{ij}=0$, where $k_{ij}$ of the $(i,j)$th entry of $K$
 and $\theta_{ij}$ is the $(i,j)$th entry of the precision matrix $\Theta=\Sigma^{-1}$. }
 Without loss of generality, we assume that $X_i \notin \bS_*$. 
 Rewrite the Fisher information matrix $I=E(b''(\bx^{T}\bbeta)\bx\bx^{T})$ 
 as $E(\bz\bz^{T})$, where $\bz=\sqrt{b''(\bx^{T}\bbeta)}\bx=\sqrt{b''(\bx_{\bS_*}^{T}\bbeta_{\bS_*})}\bx$.
 Therefore, $K=I^{-1}$ measures the partial correlation structure of $\bz$, and $k_{ij}=0$ 
 if and only if the partial correlation of $\bz_{i}$ and $\bz_{j}$ is zero. Further,
 recall that the partial correlation of  $\bz_i$ and $\bz_j$  equals zero if and only 
 if the coefficient $\tilde{c}_{ij}$ in the nodewise regression $\bz_{i}=\sum_{t \neq i}\tilde{c}_{it} \bz_{t}$ 
 equals zero, where $\{\tilde{c}_{it}: t\neq i \}$ satisfies
 \[
\{\tilde{c}_{it}: t\neq i  \}=\arg\min_{\{c_{it}: t\neq i  \}} E[\bz_{i}-\sum_{t \neq i}c_{it}\bz_{t}]^2.
 \]
 Note that for any set of Gaussian variables $\{ \bx_1,\ldots, \bx_{p_n}\}$,
\[
\{\tilde{c'}_{it}: t\neq i  \}=\arg\min_{\{c'_{it}: t\neq i  \}} E[\bx_{i}-\sum_{t \neq i}c'_{it}\bx_{t}]^2
 =\arg\min_{\{c'_{it}: t\neq i  \}} E\bigg\{[\bx_i-\sum_{t \neq i}c'_{it}\bx_{t}]^2\bigg| 
  \{\bx_1,\ldots,\bx_{p_n}\} \setminus \bx_{i}\bigg\},
 \] 
 as $E[\bx_{i}-\sum_{t \neq i}c'_{it}\bx_{t}]^2$ is minimized if and only if 
 $\sum_{t \neq i}c'_{it}\bx_{t}=E[\bx_i|  \{\bx_1,\ldots,\bx_{p_n}\} \setminus \bx_{i} ]$. 

 Since we have assumed $\bS_* \subseteq \{t:t\neq i\}$, it is easy to verify that
\begin{eqnarray*}
&&\arg\min_{\{c_{it}: t\neq i  \}} E\bigg\{[\bz_{i}-\sum_{t \neq i}c_{it} \bz_{t}]^2\bigg| 
 \{\bz_1,\ldots,\bz_{p_n}\} \setminus \bz_{i}\bigg\}\\
&=&\arg\min_{\{c_{it}: t\neq i  \}} E\bigg\{[\sqrt{b''(\bx^{T}_{\bS_*}\bbeta_{\bS_*})}\bx_{i}
  -\sum_{t \neq i}c_{it}\sqrt{b''(\bx^{T}_{\bS_*}\bbeta_{\bS_*})}\bx_{t}]^2\bigg| 
  \{\bx_1,\ldots,\bx_{p_n}\} \setminus \bx_{i}\bigg\}\\
&=&\sqrt{b''(\bx^{T}_{\bS_*}\bbeta_{\bS_*})}\arg\min_{\{c_{it}: t\neq i  \}} 
 E\bigg\{[\bx_i-\sum_{t \neq i}c_{it} \bx_t]^2\bigg| \{\bx_1,\ldots,\bx_{p_n}\} \setminus \bx_i\bigg\}
 =\{\tilde{c'}_{it}: t\neq i  \}.
\end{eqnarray*}
Since this holds for any fixed $\{\bz_1,\ldots,\bz_{p_n}\} \setminus \bz_{i}$, we further have
\[
\{\tilde{c}_{it}: t\neq i  \}=\arg\min_{\{c_{it}: t\neq i  \}} E[\bz_{i}-\sum_{t \neq i}c_{it}\bz_{t}]^2
   =\arg\min_{\{c_{it}: t\neq i  \}} E\bigg\{[\bz_{i}-\sum_{t \neq i}c_{it}\bz_{t}]^2\bigg| 
   \{\bz_1,\ldots,\bz_{p_n}\} \setminus \bz_{i} \bigg\}=\{\tilde{c'}_{it}: t\neq i  \}.
 \]
Putting these together, we have 
$k_{ij}=0\Leftrightarrow \tilde{c}_{ij}=0 \Leftrightarrow \tilde{c'}_{ij}=0 
 \Leftrightarrow \theta_{ij}=0$,
which justifies the proposition.

Next, we partition $K$ as
\[
K=\begin{bmatrix}
K_{d}     & K_{d,p-d} \\
K_{p-d,d}       & K_{p-d}  
\end{bmatrix}.
\]
By the formula of partitioned matrix inverse, 
 we have $I_{d}=(K_{d}-K_{d,p-d}K_{p-d}^{-1}K_{p-d,d})^{-1}$ and $I_{d}^{-1}=K_{d}-K_{d,p-d}K_{p-d}^{-1}K_{p-d,d}$.
 Since the true features set $\bS_*$ has been included in 
 $\{1,\dots,d\}$, we can apply the previous proposition to obtain that $k_{1j}=0\Leftrightarrow \theta_{1j}=0$ 
 for $j>d$.  Since the Markov neighborhood of $\bx_1$ has been included in $\{1,\dots,d\}$, 
 we have $\theta_{1j}=0$ for $j>q$. Combining these two together, 
 we have $k_{1j}=0$ for $j>q$; in other words, the first row of $K_{d,p-d}$ and first column of $K_{p-d,d}$ are exactly 
 zero. Therefore, the $(1,1)$-th element of  $K_{d,p-d}K_{p-d}^{-1}K_{p-d,d}$ is also exactly zero, 
  and the $(1,1)$-th entry of $K_{d}$ equals to the $(1,1)$-th entry of $I_{d}^{-1}$, i.e.,
 $k_{11}=k_{d,11}$. This completes the proof of part (i). 

 Since for the GLMs, including the logistic, Poisson and Cox regression, 
 $J_n(\bbeta)$ is almost surely Lipschitz continuous, $J_n(\bbeta)$ is 
 stochastically equicontinuous.  Therefore,
\[
 |J_n(\hat{\bbeta}_{D_j}) -I(\bbeta_{D_j}) | \leq |J_n(\hat{\bbeta}_{D_j})-J_n(\bbeta_{D_j})|
 +|J_n(\bbeta_{D_j})-I(\bbeta_{D_j})| 
 \stackrel{p}{\to} 0,
\]
 where $|J_n(\hat{\bbeta}_{D_j})-J_n(\bbeta_{D_j})|\stackrel{p}{\to} 0$ follows from 
 the stochastic equicontinuity of $J_n(\cdot)$ and the consistency of the MLE $\hat{\bbeta}_{D_j}$, and 
 $|J_n(\bbeta_{D_j})-I(\bbeta_{D_j})|\stackrel{p}{\to} 0$ follows from the weak law of large numbers.  
 This completes the proof of part (ii).
\end{proof}

\subsection{A Variable Screening-Based MNR Algorithm}
\label{SMsect2}

\subsubsection{The Algorithm}

\begin{algorithm} (Variable screening-based Markov neighborhood regression) \label{subsetAlg2}
\begin{itemize}

 \item[(a)] (Variable screening) 
    Apply the sure independence screening procedure \citep{FanLv2008} 
    to obtain a reduced feature set, $\hat{\bS}_* \subseteq \{1,\dots,p\}$, with 
    the size $|\hat{\bS}_*|=O(\sqrt{n}/\log(n))$.

    \item[(b)] (Markov blanket estimation) 
       Apply the correlation screening procedure \citep{Song2014} 
       to $\bX$ to obtain a reduced neighborhood 
       $\hat{\xi}_{j}\subseteq \{1,\dots,p\}$ for each variable $X_j$ 
       with the size $|\hat{\xi}_j|=O(\sqrt{n}/\log(n))$.

        \item[(c)] (Subset Regression) For each variable $X_j$, $j=1,\ldots,p$, run the OLS regression 
  with the features given by $\{X_j\} \cup \bX_{\hat{\xi}_j} \cup \bX_{\hat{\bS}_*}$. 
  Conduct inference for $\beta_j$, including the estimate, confidence interval and $p$-value,
  based on the output of the subset regression. 
\end{itemize}
\end{algorithm} 

The variable screening step restricts the size of $\hat{\bS}_*$ to $O(\sqrt{n}/\log(n))$, which looks more restrictive than
the order $O(n/\log(n))$ used by conventional variable screening algorithms. However, this is just a technical condition
and will not affect much on the the actual size of $\hat{\bS}_*$. 
For the step of Markov blanket estimation, this is similar. 
To justify the validity of Algorithm \ref{subsetAlg2}, we establish Theorem \ref{Them2} with part of the conditions given as follows: 

\begin{enumerate}
       \item[(B1)] The distribution $P_{\bX}$ is multivariate Gaussian,
         and it satisfies the Markov property
         with respect to the undirected graph $\bG$ for all size of $\bX$.

       \item[(B2)] The correlation coefficients satisfy
         $\min \{|r_{ij}|;e_{ij}=1, i,j=1,\dots,p_{n}, i\neq j \}\geq c_{0}n^{-\kappa}$
        for some constants $c_{0}>0$ and $0<\kappa<1/4$, and
        $\max \{|r_{ij}|; i,j=1,\dots,p_{n}, i\neq j \}\leq M_{r} < 1$ for some constants $0<M_{r}<1$.

        \item[(B3)] There exists constants $c_{1}>0$ and $0\leq \tau < \frac{1}{2}-2\kappa$
                      such that $\lambda_{\max}(\Sigma)\leq c_{1}n^{\tau}$.
\end{enumerate}

 \begin{theorem} \label{Them2} (Validity of Algorithm \ref{subsetAlg2})
   If the conditions (A0), (A6), (A9) (given in the Appendix of the main text), and (B1)-(B3) hold, 
   the sure independence screening (SIS) algorithm \citep{FanLv2008}
   is used for variable selection in step (a),  and 
   the correlation screening algorithm \citep{Song2014} is used for estimating $\xi_j$'s 
   in step (b),  
   then for each $j \in \{1,2,\ldots,p_n\}$, 
  $\sqrt{n} \frac{\hat{\beta}_j-\beta_j}{\sqrt{ \hat{\sigma}_n^2 \hat{\theta}_{jj}}}  \sim N(0,1)$ as $n\to\infty$, where $\hat{\beta}_j$ denotes the estimate of $\beta_j$ obtained 
   from the subset regression,
   $\hat{\sigma}_n^2=(\by-\bx_{D_j}\hat{\bbeta}_{D_j})^T(\by-\bx_{D_j}\hat{\bbeta}_{D_j})/(n-d-1)$,
  $\hat{\theta}_{jj}$ is the $(j,j)$-th entry of
  the matrix $\bigg[\frac{1}{n} \sum_{i=1}^{n}\bx^{(i)}_{D_j}(\bx^{(i)}_{D_j})^{T} \bigg]^{-1}$, and
  $\bx^{(i)}_{D_j}$ denotes the $i$-th row of $\bX_{D_j}$.
\end{theorem} 
 \begin{proof} 
  With conditions (A0) and (B1)-(B3), we can apply Theorem 1 and Theorem 2 of \cite{Song2014}
 to obtain that $P(\hat{\xi}_{j}\supseteq \xi_{j})=1-o(1)$, for $j=1,\dots,p_{n}$. 
 In addition, $P(|\hat{\xi}_j|\leq o(n^{\frac{1}{2}}))=1-o(1)$ holds with   
 an appropriate threshold of correlation coefficients. 
 Similarly, with the conditions (A0), (A6), and (B3), we can apply Theorem 1 of \cite{FanLv2008}
 to obtain that $P(\hat{\bS}_* \supseteq \bS_*)=1-o(1)$. 
 According to Theorem 1 of \cite{FanLv2008}, the size of $\hat{\bS}_*$ can be chosen in the order
 of $O(n^{1-\vartheta})$ for some $\vartheta<1-2\kappa-\tau$, i.e., $1-\vartheta> 2\kappa+\tau$.   
 While, by (B3),  we have $2\kappa+\tau=\frac{1}{2}-\varepsilon <1/2$ for some 
 $\varepsilon>0$. Therefore, we can choose $1-\vartheta=\frac{1}{2}-\frac{\varepsilon}{2}$, 
 which implies $|\hat{\bS}_*|=o(n^{1/2})$. As suggested in the algorithm, we can choose 
 the size of $\hat{\bS}_*$ in the order $O(\sqrt{n}/\log(n))$. 
 Putting these together, we have $P(|D_j|=|\{j\}\cup\hat{\xi}_{j}\cup \hat{\bS}_*|=o(n^{1/2}))=1-o(1)$.
 By condition (A9), we can get the consistency of $\hat{\sigma}_n^2$
 based on Theorem 1 of \cite{FanGH2012} and the asymptotic $P(|D_j|\log(p_n)=o(n))=1-o(1)$.
 Finally, by applying Lemma 1 and Slutsky's theorem, 
 we can conclude the proof of the theorem. 
\end{proof} 

Algorithm \ref{subsetAlg2} can have many variants.
For example, we can replace step (a) by Lasso, which has been empirically found to perform rather well
for screening in comparison to marginal correlation screening, 
see \cite{BuhlmannM2014} for empirical comparisons. We can also replace step (a) by 
the Henze-Zirkler sure independence screening (HZ-SIS) algorithm \citep{XueLiang2017}.  
 
 HZ-SIS is a model-free algorithm, which is developed based on Henze-Zirkler's multivariate normality test \citep{HenzeZ1990}.  When applying HZ-SIS to Gaussian data, the nonparanormal transformation \citep{LiuH2009} step might be omitted, but including such a step, 
 which provides an appropriate truncation for extreme values,  might robustify the performance of the algorithm. For Gaussian data, this algorithm is equivalent to test whether the correlation coefficient is zero. Since HZ-SIS possesses the ranking consistency property, which stems from the consistency of the Henze-Zirkler's test and is 
 established in Theorem \ref{RCPthem} below,
 it is expected to perform better than the SIS algorithm by \cite{FanLv2008}.  
 The ranking consistency property is a stronger property than sure screening,
 which enables an asymptotic separation of the cases 
 of zero correlation and non-zero correlation. Further, under condition (C4), we 
 will be able to restrict the size of $\hat{\bS}_*$ to $O(\sqrt{n}/\log(n))$.  


\paragraph{Ranking Consistency of the HZ-SIS Algorithm:} 
 Let $D=\min\{k:$ $Y$ and $X_k$ are marginally dependent$\}$, and 
 let $I=\min\{k:$ $Y$ and $X_k$ are marginally independent$\}$. 
 Let $\omega_k$ denote the screening statistic for the variable $X_k$, which is defined by
 \[
 \omega_k=\int_{\mR^2} \left| \psi_k(\bt)-\exp(-\frac{1}{2}\bt' \bt) \right| \phi_{\beta}(\bt) d \bt,
 \]
 where $\phi_{\beta}(\bt)$ is the PDF of $N(0,\beta^2 I_2)$,  $\beta$ is a smooth parameter with the
 optimal value $(1.25n)^{1/6}/\sqrt{2}$; and 
 $\psi_k(\bt)$ is the characteristic function of $(\Phi^{-1}(F_k(X_k)), \Phi^{-1}(F_y(Y)))$, 
 $\Phi$ denotes the CDF of the standard normal 
 distribution, and $F_k$ and $F_y$ denote the CDFs of $X_k$ and $Y$, respectively. 
 
\begin{enumerate} 
 \item[(C1)] There exist positive constants $c>0$ and $0\leq \kappa \leq 1/4$ such that 
   $\min_{k\in D} \omega_k \geq 2 cn^{-\kappa}$.

 \item[(C2)] The dimension $p=O(\exp(n^{\tau}))$ for some constant $0 \leq \tau < (1-4\kappa)/3$. 
 
 \item[(C3)] $\lim\inf_{p\to \infty}\{\min_{k\in D} \omega_k-\max_{k \in I} \omega_k\} \geq c_3$, 
             where $c_3>0$ is a constant. 
\item[(C4)] $|D|=o(\sqrt{n})$.
\end{enumerate} 
 
 Since the Henze-Zirkler test is consistent \citep{HenzeZ1990}, 
 for which the power tends to 1 for all cases in $D$, 
 (C3) is a relatively weaker assumption. Let $\hat{\omega}_k$ denote the Henze-Zirkler test 
 statistic, i.e., an estimator of $\omega_k$,   
 calculated with $\beta=(1.25 n)^{1/6}/\sqrt{2}$ and $n$ observations, and let $\hat{D}_n$ denote 
 the estimate of $D$. Under conditions (C1) and (C2), 
 \cite{XueLiang2017} proved the screening property $P(D \subseteq \hat{D})=1-o(1)$ 
 as $n\to \infty$. Under conditions (C1)-(C3), the ranking consistency property 
 can be established as follows: 
 
 \begin{theorem} \label{RCPthem} If conditions (C1)-(C3) hold, then 
  $\lim\inf_{n\to \infty} \{ \min_{k\in D} \hat{\omega}_k-\max_{k \in I} \hat{\omega}_k \}>0$, a.s. 
 \end{theorem}
 \begin{proof} 
  The theorem can be proved in a similar way to Theorem 2.2 of \cite{CuiLZ2015}. Calculate the probability  
  \[
  \begin{split} 
  & P\{ \min_{k\in D} \hat{\omega}_k -\max_{k \in I} \hat{\omega}_k < c_3/2\} \leq 
    P\{ (\min_{k\in D} \hat{\omega}_k -\max_{k \in I} \hat{\omega}_k)-
        (\min_{k\in D} {\omega}_k -\max_{k \in I} {\omega}_k) < -c_3/2 \} \\ 
  & \leq  P\{ |(\min_{k\in D} \hat{\omega}_k -\max_{k \in I} \hat{\omega}_k)-
        (\min_{k\in D} {\omega}_k -\max_{k \in I} {\omega}_k)| > c_3/2 \}  
    \leq P\{ 2\max_{1 \leq k \leq p} |\hat{\omega}_k-\omega_k|>c_3/2 \} \\
  & \leq O(p \exp(-c_1 n^{\frac{1-4\kappa}{3}})), \\
  \end{split}
  \]
 by letting $c_3=n^{-\kappa}$, where the last inequality follows from Lemma 1 of 
 \cite{XueLiang2017}. Then, for some $n_0$, 
 $\sum_{n=n_0}^{\infty} p \exp(-c_1 n^{\frac{1-4\kappa}{3}}) \leq \sum_{n=n_0}^{\infty} n^{-2} <  \infty$. 
 Therefore, we obtain that $\lim\inf_{n\to \infty} \{\min_{k\in D} \hat{\omega}_k-\max_{k\in I} \hat{\omega}_k\} 
  \geq c_3/2>0$ a.s. 
 \end{proof}

 \subsubsection{Illustration of Algorithm \ref{subsetAlg2}}  \label{Alg2Sect}
 
  Algorithm \ref{subsetAlg2} was also applied to the same 100 datasets simulated in Section 3.2 of the main text. 
  For simplicity, we set the same upper bound for the Markov blanket size $|\hat{\xi}_j|$ 
  and the model size $|\hat{\bS}_*|$. Let $m$ denote the upper bound. 
  Table \ref{screentab} 
  summarizes the outputs of the algorithm with different values 
  of $m=3$, 5, 8, 15 and 20.  
  The results indicate that Algorithm \ref{subsetAlg2} performs well for this example. 
  Although the widths of the resulting confidence intervals vary with the value of $m$, the 
  coverage rates are pretty stable, which are always close to the nominal 
  level 0.95 for both zero and nonzero regression coefficients.
  In the case $m=3$, where some true features are missed in each subset regression, 
  Algorithm \ref{subsetAlg2} still performs reasonably well, although the resulting 
  confidence intervals are a little wider.  
  When $m=8$, the algorithm attains its  best performance. Of course, such a choice of $m$ depends on how the data was generated. 
Intuitively, the best choice of $m$ should be the smallest size such that all true features can be covered by a reduced feature set of that size and for almost every variable, its Markov blanket can be covered by a reduced neighborhood of that size.

 \begin{table}[htbp] 
\caption{Coverage rates and widths of the 95\% confidence intervals 
produced by Algorithm \ref{subsetAlg2} for the Toeplitz-covariance linear regression 
with different values of $m$, which controls the size of Markov neighborhoods. 
Refer to the caption of Table 1 of the main text for the notation. }
 \vspace{-0.2in}
\label{screentab}
\begin{center}
\begin{tabular}{ccccccc} \toprule
   Measure      &  &  $m=3$    & $m=5$    & $m=8$   & $m=15$    & $m=20$ \\ \midrule
  & signal  &  0.956(0.021) & 0.962(0.019) &  0.956(0.021) &  0.958(0.020) & 0.954(0.021) \\
 \raisebox{1.5ex}{Coverage} & noise  
            &  0.952(0.021) & 0.950(0.022) &  0.951(0.022) &  0.951(0.022) & 0.949(0.022) \\ \midrule
  & signal  &   1.023(0.035) & 0.854(0.014) &  0.839(0.011) &  0.857(0.011) & 0.876(0.011) \\
\raisebox{1.5ex}{Width}         & noise   &   2.566(0.019) & 1.610(0.054) &  0.902(0.008) &  0.935(0.007) & 0.963(0.008) \\ \bottomrule
\end{tabular}
\end{center}
\end{table}
  
  A comparison of Table 1 (of the main text)
  and Table \ref{screentab} shows that the confidence intervals produced by 
  Algorithm \ref{subsetAlg2} tend to be wider than those by Algorithm 1. This is reasonable, as 
  Algorithm 1 seeks to use the smallest Markov neighborhood for each feature 
  and thus the resulting inference is more accurate. 
  However, as shown in Table \ref{timetab2}, Algorithm \ref{subsetAlg2} is much more efficient than Algorithm 1 in computation. 
  Algorithm \ref{subsetAlg2} is an accuracy/efficiency trade-off version of Algorithm 1.

\subsubsection{Computational Complexity} 

In Algorithm \ref{subsetAlg2}, SIS was used in both variable screening and Markov blanket estimation steps. 
By \cite{FanLv2008}, the computational complexity of the variable screening step is $O(np)$. Therefore, the computational complexity of the Markov blanket estimation step is $O(np^2)$, as 
SIS needs to be done for each of $p$ features. By the size of the Markov neighborhood, 
the computational complexity of the subset regression step is $O(n^2 p)$. Therefore, the total computational complexity of Algorithm \ref{subsetAlg2} is $O(np^2)$ for small-$n$-large-$p$ problems. In contrast, as analyzed in Section 3.5 of the main text, the computational complexity of Algorithm 1 is $O(n^3p^2)$.

For the dataset used in Section 3.5 of the main text, 
Algorithm \ref{subsetAlg2} cost only 3.9 seconds when running with a single thread on an Intel(R) Xeon(R) CPU E5-2660 v3@2.60GHz machine.
Table \ref{timetab2} compares the CPU time cost by
different methods for the same dataset on the same machine, 
where the parts for desparsified-Lasso and Algorithm 1 are taken from Table 8 of the main text. 

\begin{table}[htbp]
\begin{center}
\caption{CPU times (in seconds) cost by different methods for 
a dataset with $n=200$ and $p=500$, 
 where MNR$_a$, MNR$_b$, MNR$_c$, and MNR$_d$ represent different implementations of Algorithm 1 as
 explained in Table 8 of the main text.}
\label{timetab2}
\begin{tabular}{ccccccc} \toprule
          &     & \multicolumn{4}{c}{Algorithm 1} &  \\ \cline{3-6} 
Methods  &  Desparsified-Lasso &    MNR$_a$ & MNR$_b$ & MNR$_c$ & MNR$_d$ & Algorithm \ref{subsetAlg2} \\ \midrule
CPU(s)   &   258               &    152     &  230    & 205  &  250  & 3.9 \\ \bottomrule
\end{tabular}
\end{center}
\end{table}

\subsection{Figures Illustrating the Performance of MNR in Variable Selection} 
\label{SMsect3}

\

 \begin{figure}[htbp]
\centering
\begin{center}
\begin{tabular}{c}
\epsfig{figure=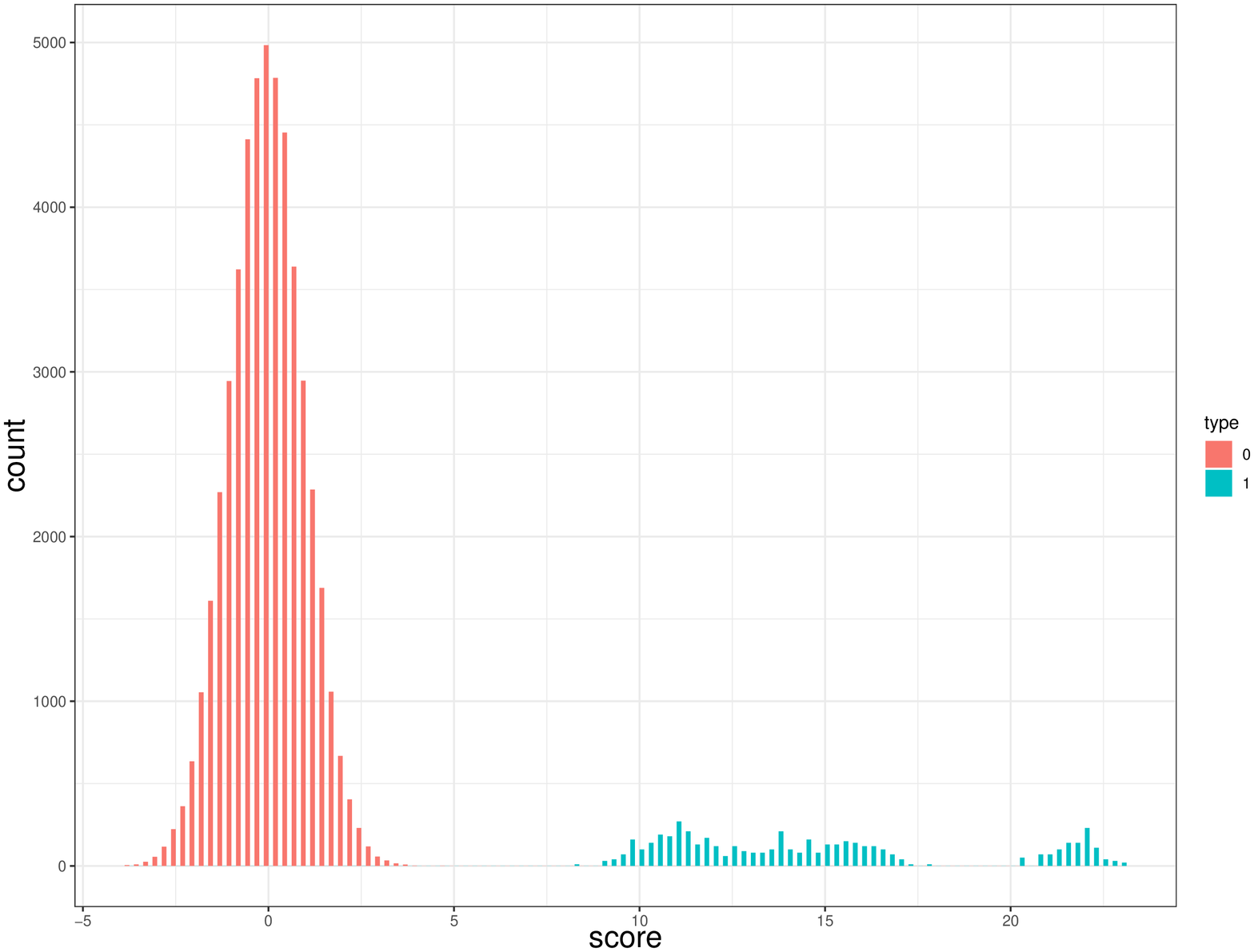,height=5.0in,width=2.5in, angle=270} 
\end{tabular}
\end{center}
\caption{Histogram of the z-scores produced by MNR for the Toeplitz-covariance linear regression example, 
where the z-scores of  the true and false features are shown in blue and red, respectively. The z-scores of the true 
 features have been replicated 10 times for a better view of the plot. }
\label{histsela}
\end{figure}

 \begin{figure}[htbp]
\centering
\begin{center}
\begin{tabular}{c}
\epsfig{figure=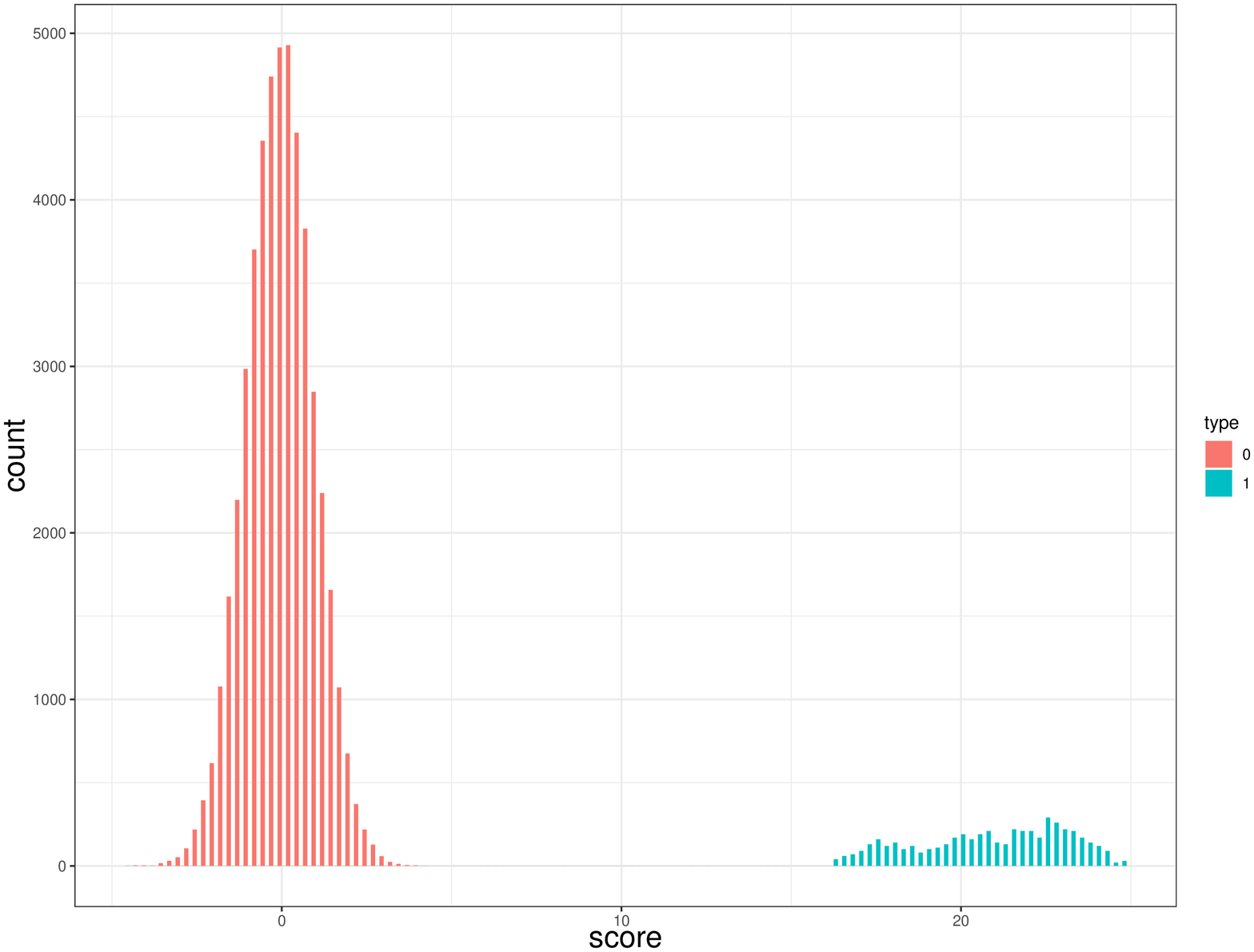,height=5.0in,width=2.5in,angle=270}  
\end{tabular}
\end{center}
\caption{Histogram of the z-scores produced by MNR for the AR(2)-precision linear regression example, where 
the z-scores of  the true and false features are shown in blue and red, respectively. The z-scores of the true 
 features have been replicated 10 times for a better view of the plot. }
\label{histselb}
\end{figure}

\begin{figure}[htbp]
\centering
\begin{center}
\begin{tabular}{c}
\epsfig{figure=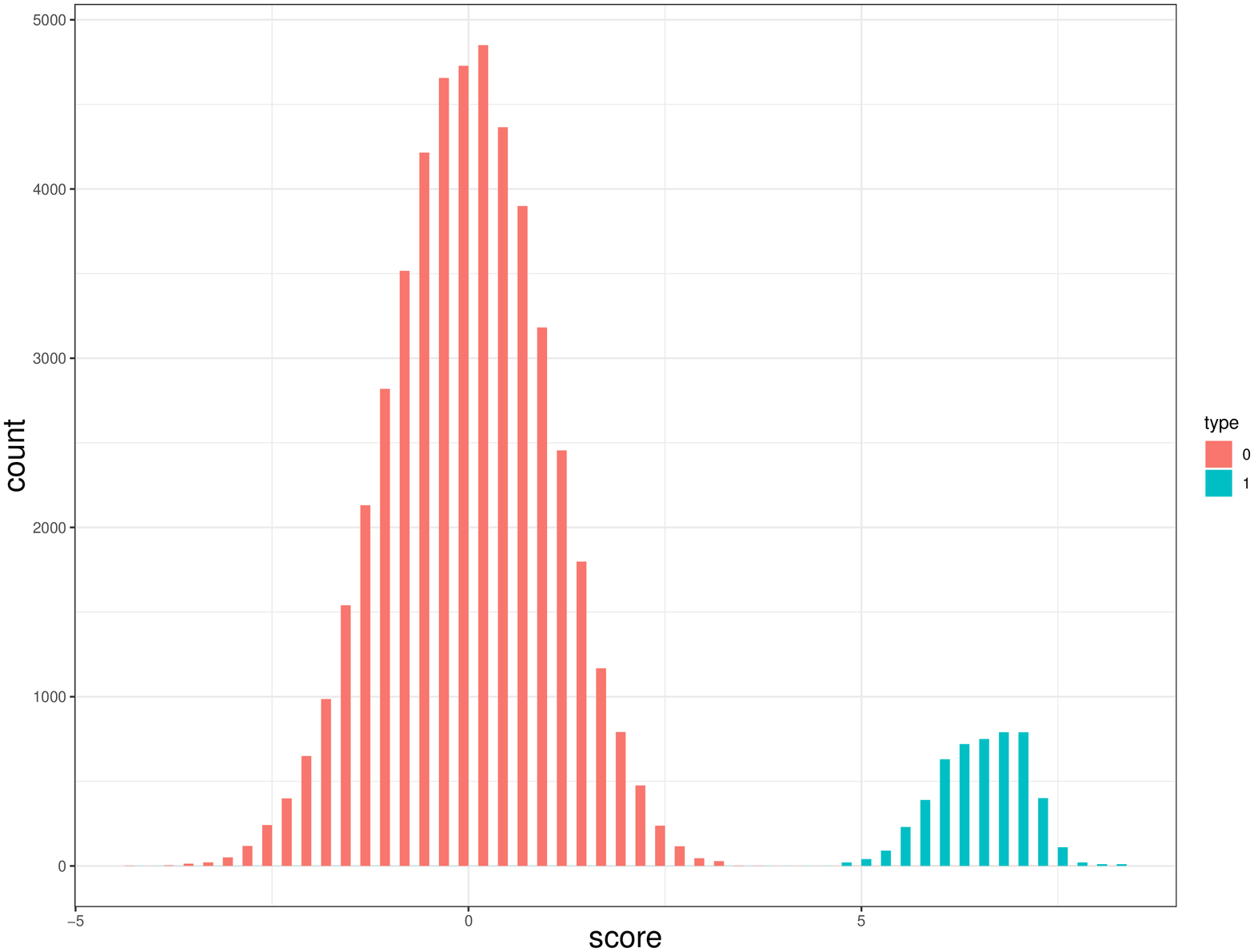,height=5.0in,width=2.5in,angle=270} 
\end{tabular}
\end{center}
\caption{Histogram of the z-scores produced by MNR for the AR(2)-precision logistic regression example,  where
the z-scores of  the true and false features are shown in blue and red, respectively. The z-scores of the true features have been replicated 10 times for a better view of the plot. }
\label{histselc}
\end{figure}

 \begin{figure}[htbp]
\centering
\begin{center}
\begin{tabular}{c}
\epsfig{figure=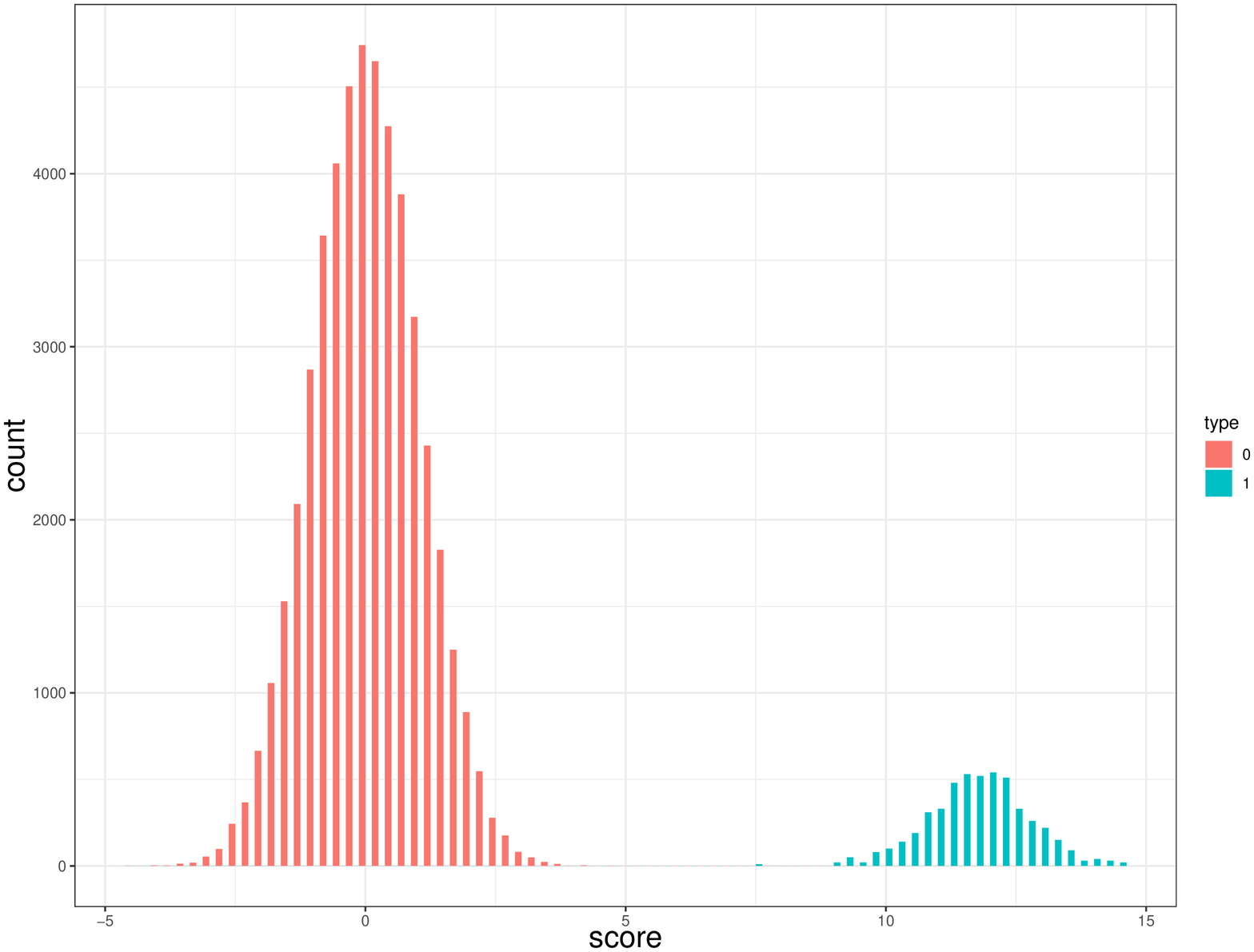,height=5.0in,width=2.5in,angle=270}
\end{tabular}
\end{center}
\caption{Histogram of the z-scores produced by MNR for the AR(2)-precision Cox regression example, where 
the z-scores of  the true and false features are shown in blue and red, respectively. The z-scores of the true  features have been replicated 10 times for a better view of the plot. }
\label{histseld}
\end{figure}

 \newpage

\bibliography{ref}
\end{document}